\makeatletter
\newcommand*{\rom}[1]{\expandafter\@slowromancap\romannumeral #1@}
\makeatother
\documentclass[12pt, draftclsnofoot, onecolumn]{IEEEtran}
\IEEEoverridecommandlockouts
\ifCLASSINFOpdf
\else
\fi
\usepackage{graphics}
\usepackage{multicol}
\usepackage{lipsum}
\usepackage{color}
\usepackage[cmex10]{amsmath}
\usepackage{amsthm}
\usepackage{amssymb}
\usepackage{mathrsfs}
\usepackage{mathtools}
\usepackage{amsbsy}
\usepackage{mathrsfs}
\usepackage{setspace}
\usepackage{float}
\usepackage{graphicx}
\usepackage{psfrag}
\usepackage{epstopdf}
\usepackage{lettrine}
\usepackage{psfrag}
\usepackage{newclude}
\usepackage[normalem]{ulem}
\usepackage{latexsym}

\usepackage[table]{xcolor}
\usepackage{multirow}
\usepackage{rotating}
\usepackage{booktabs}
\usepackage{amsmath}
\usepackage{wasysym}
\usepackage{mathrsfs}
\usepackage{bbm} 
\usepackage{filecontents}
\usepackage{amsmath,epsfig,cite,amsfonts,amssymb,psfrag,color}

\usepackage{apptools}

\usepackage{chngcntr}
\ifCLASSOPTIONcompsoc
\usepackage[caption=false,font=normalsize,labelfon
t=sf,textfont=sf]{subfig}
\else
\usepackage[caption=false,font=footnotesize]{subfig}
\fi
\AtAppendix{\counterwithin{lemma}{section}}

\usepackage{accents}
\newlength{\dhatheight}

\allowdisplaybreaks
\newtheorem{remark}{Remark}
\newtheorem{theorem}{Theorem}

\newtheorem{proposition}{Proposition}

\newtheorem{Remark}{Remark}
\graphicspath{{figures/}}
\allowdisplaybreaks

\begin{document}
\title{Transmission Strategies for Cell-Free Massive MIMO with Limited-Capacity Fronthaul Links}

\author{\IEEEauthorblockN{\normalsize
    Hamed Masoumi and Mohammad Javad Emadi}
  \IEEEauthorblockA{\\ Electrical Engineering Department, Amirkabir University of Technology, Tehran, Iran\\
    E-mails: \{hamed\_masoomy, mj.emadi\}@aut.ac.ir}
    \vspace{-.75cm}
    }
\maketitle
\begin{abstract}
We study an uplink scenario of a cell-free massive multiple-input multiple-output (CF-mMIMO) system with limited capacity fronthaul links (LC-FHLs) connecting each access point (AP) to central unit (CU), where user equipments and APs are subject to hardware impairments. Therefore, to efficiently use the capacity of FHLs to maximize the achievable rate, we analyze three strategies for performing compression and forwarding of channel state information (CSI) and data signals over the LC-FHLs to the CU; Compress-forward-estimate (CFE), estimate-compress-forward (ECF), and estimate-multiply-compress-forward (EMCF). For CFE and EMCF achievable rates are derived, and for ECF one upper and lower bounds are presented which are tight for ideal hardwares and FHLs. Also for forwarding the quantized version of CSI and data signals of each user, low-complexity fronthaul capacity allocations are proposed for ECF and EMCF strategies, which considerably improve the performance of the system especially for limited capacity FHLs. Our results indicate that at high SNR regime and for large enough capacity of FHLs, estimating channels at the CU rather than APs result in smaller estimation error. Then, geometric programming power allocations are developed for CFE and ECF to maximize sum rates. Finally, to highlight performance characteristics of the system numerical results are presented.
\end{abstract}
\IEEEpeerreviewmaketitle

\section{Introduction}
\lettrine{M}{ASSIVE} MIMO is a key technology to answer the growing demands in the next generation of wireless networks. It offers high improvements in spectral and energy efficiencies (SE and EE), and accompanies near-optimal linear processing owing to the weak law of large numbers \cite{marzetta2010noncooperative,ngo2013energy}. Utilizing massive number of antennas distributed in an area and serving a smaller number of user equipments (UEs) in same time-frequency resources named as the CF-mMIMO system \cite{ngo2017cell}. Thus, thanks to the macro-diversity owing to distributed antennas, one can provide uniformly good services to all the UEs. However this achievement comes at the price of increased fronthaul data load and deployment cost.

CF-mMIMO system is studied in the literature from different perspectives. For instance, uplink (UL) and downlink (DL) achievable rates with maximum ratio combining (MRC) and conjugate beamforming (CB) are considered in seminal paper \cite{ngo2017cell}, and optimal power allocations are also proposed to maximize the minimum of users' rates. \cite{nayebi2017precoding} considers max min power control for the DL transmission with CB and zero-forcing (ZF) precodings. By applying CB and ZF precodings for DL scenario, EE is studied respectively in \cite{ngo2018total,nguyen2017energy}. Furthermore, CF-mMIMO system with;  Non-orthogonal multiple access \cite{li2018noma}, user-centric approach to minimize fronthaul load \cite{buzzi2017cell}, pilot power allocation to minimize channel estimation error \cite{mai2018pilot}, joint power control and load balancing using ZF and MRC techniques \cite{nguyen2018optimal}, and channel hardening and favorable propagation analysis via stochastic geometry \cite{chen2018channel}, are investigated in the literature.

In practice, communication networks suffer from number of non-idealities such as, carrier-frequency and sampling-rate offset, IQ-imbalance, phase-noise and non-linearity of analog devices. These issues can be partially compensated by use of high-quality hardwares or more complex processing which result in expensive and high power consuming implementations. However, there would still remain some non-negligible errors named as residual hardware impairments (HI) \cite{studer2010mimo}. For mMIMO setup, utilization of low-cost hardwares is essential to deploy a cost-efficient system. Therefore, analyzing effects of the non-negligible distortions caused by HI at transceivers on  performance of mMIMO system has gained research interests \cite{bjornson2014massive,bjornson2017massive,zhang2018performance,zhang2017spectral,bjornson2015massive,zhang2018spectral,zhu2017analysis,zhang2018scaling,papazafeiropoulos2018impact}. It is shown that utilization of non-ideal hardware in mMIMO system is the limiting factor for SE/EE and channel estimation accuracy particularly at high SNR \cite{bjornson2014massive}. Recently, for a CF-mMIMO system with MRC scheme, max-min power allocation in presence of HIs is studied in \cite{zhang2018performance}. \cite{bjornson2015massive} analyzes SE for UL transmission of a multi-cell mMIMO system, where only the receivers are subject to HIs. Impact of the HI on SE of a multi-pair two-way massive-antenna relay subject to HI only at the relay is studied in \cite{zhang2018spectral}. Moreover in presence of HIs, secrecy performance of the mMIMO systems \cite{zhu2017analysis}, scaling behavior of rate in a multi-cell mMIMO system with Rician fading \cite{zhang2018scaling}, and performance of the Rayleigh-product large MIMO channels \cite{papazafeiropoulos2018impact}, are investigated as well.

On the other hand, analyzing effects of limited capacity FHLs connecting each AP to the central processing unit is of interest from practical points of view. In \cite{kang2014joint}, a conventional multi-user MIMO system with limited capacity FHLs and perfect hardwares are analyzed and two  joint signal and CSI compression schemes for forwarding over the FHLs are studied. Recently, a CF-mMIMO system with limited capacity FHLs while assuming \emph{perfect hardwares} at UEs and APs, is studied in \cite{bashar2018cell}. For equal fronthaul capacity allocation for forwarding quantized CSI and data signals to the CU, optimal power allocation to maximize the minimum achievable data rates of UEs is investigated. Besides, compute-and-forward strategy is studied for a CF-mMIMO to decrease the data load on the FHLs \cite{huang2017compute}.

It worth mentioning that, in contrast to the previous works, this paper \emph{jointly} analyses the HIs and limited capacity FHLs on the achievable rate of a CF-mMIMO system and studies optimal resource allocations to maximize SE. Because of LC-FHLs, we investigate three strategies for joint data signal and CSI compression at APs, which are called CFE, EMCF and ECF, where EMCF and CFE have the highest and lowest computational complexities at the APs, respectively, and EMCF achieves the highest data rate. For each strategy, fronthaul capacity allocation for forwarding CSI and data signal to the CU is discussed, and data power control optimization problem to further improve the sum-rate of the system is investigated for CFE and ECF strategies. The main contributions of this paper are summarized as follows,

\begin{itemize}
  \item Considering UL channel training and data transmission of a CF-mMIMO systems wherein APs communicate with CU over LC-FHLs and UEs and APs are subject to HIs.
  \item Minimum-mean-square-error (MMSE) estimations of the channels are obtained at the APs and then the quantized CSI forwarded to the CU over the LC-FHLs, or the received pilot signals at the APs are quantized and forwarded to the CU to apply MMSE estimation.
  \item Based on the above estimation methods and whether the quantized version of the data signal or the quantized version of the weighted data signal be available at the CU, achievable rates are derived for three strategies CFE, EMCF, and ECF. Closed-form expressions for the achievable rate of CFE and EMCF are derived, and upper and lower bounds for the achievable data rate of ECF are obtained which are tight for perfect transceiver hardwares and large enough capacity of FHLs. It is also proved that at high SNR regime and large enough fronthaul capacity, the CFE strategy can result in lower channel estimation error than that of ECF.
  \item Fronthaul capacity allocation for CSI and data signal of each user is discussed. We propose low-complexity fronthaul capacity allocation solutions for the ECF and EMCF based on the \emph{water-filling} notion, since the original optimization problems are analytically untractable.
  \item It is proved that for a single-user CF-mMIMO scenario at high SNR regime, the ECF strategy outperforms the CFE one.
  \item Data power optimization problem to maximize the sum rates of CFE and ECF strategies are studied, and the original optimization problems are approximated with geometric programming.
  \item It is shown that for some special cases, e.g. perfect HI and or infinite-capacity FHLs, our results include that of previously obtained in the literature.
  \item Finally, numerical results are provided to illustrate and compare the performance (e.g. sum SE, EE, channel estimation quality) of the different strategies under various parameters (e.g. capacity of FHL, hardware quality, number of UEs and APs).
\end{itemize}

\emph{Organization}: In Section \ref{sec2sysmodel}, we introduce the system model for CF-mMIMO with transceiver HI and LC-FHLs. UL achievable data rates for different strategies, i.e. CFE, ECF, and EMCF, with transceiver HIs and LC-FHLs are derived in Section \ref{sec3performance}. Section \ref{sec4power} deals with the power allocation problem and Section \ref{sec5numerical} presents numerical results. Finally, the paper is concluded in Section \ref{sec6conclude}.

\emph{Notation}: $\boldsymbol{x}\in \mathbb{C}^{N\times 1}$ denotes a vector in a $N$ dimensional complex space. $\delta(n)$ indicates the Dirac delta function which is 1 for $n=0$, otherwise it is zero. The statistical expected value of a random variable and the conditional one are denoted by $\mathbb{E}\{x\}$, and $\mathbb{E}\{x|y\}$, respectively. The standard information theoretic measures such as differential entropy, mutual information, and the conditional one are respectively denoted by $h(x)$, $I(x;y)$, and $I(x;y|z)$.

\section{System Model}\label{sec2sysmodel}
We consider uplink data transmission of a CF-mMIMO system where $M$ single-antenna APs are distributed in a given area and simultaneously serve $K$ single-antenna UEs, as depicted in Fig.\ref{fig:sysmodel}. AP\textsubscript{$m$}, where $m\!\in\! \{1, 2, ..., M\}$, is connected to a CU via a FHL with limited capacity $C_{m}$ [bits/s/Hz]. Moreover, it is assumed that the wireless channels between APs and UEs are modeled as block fading ones with coherence time $T_C$ [s] and coherence bandwidth of $B_{C}$ [Hz]. Thus, the coherence interval in denoted by $T \!=\! \lfloor T_{C}B_{C}\rfloor$ [samples], and the channel between each AP-UE pair over a coherence interval is modeled by
\begin{equation}\label{eqn:Channel}
      g_{mk} =\sqrt{\beta_{mk}}h_{mk}‎, ‎\text{~for~} m=1,...,M, \text{~and ~} k =1,...,K,
\end{equation}
where, $\beta_{mk}$ denotes the large-scale fading, and $\{h_{mk}\}_{m,k}$ represent the small-scale fading coefficients which are modeled by independent and identically distributed (i.i.d) zero-mean and unit variance complex Gaussian random variables, that is $h_{mk}\sim \mathcal{CN}(0, 1)$.

Before we dive into performance analysis of uplink transmission (channel training and data transmission) in presence of HIs and FHLs with limited capacities, let us briefly explain technical preliminaries in the following subsections.
\begin{figure}[!t]
\centering
\psfrag{UE1}[][][0.6]{\textcolor[rgb]{0.10,0.5,1.00}{\ \ UE\textsubscript{$1$}}}
\psfrag{UE2}[][][0.6]{\textcolor[rgb]{0.10,0.5,1.00}{\ \ UE\textsubscript{$2$}}}
\psfrag{UE3}[][][0.6]{\textcolor[rgb]{0.10,0.5,1.00}{\ \ UE\textsubscript{$3$}}}
\psfrag{UEk}[][][0.6]{\textcolor[rgb]{0.10,0.5,1.00}{\ \ UE\textsubscript{$k$}}}
\psfrag{UEK}[][][0.6]{\textcolor[rgb]{0.10,0.5,1.00}{\ \ \ UE\textsubscript{$K$}}}
\psfrag{AP1}[][][0.6]{\textcolor[rgb]{0.1016,0.6523,0.1016}{AP\textsubscript{$1$}}}
\psfrag{AP2}[][][0.6]{\textcolor[rgb]{0.1016,0.6523,0.1016}{AP\textsubscript{$2$}}}
\psfrag{AP3}[][][0.6]{\textcolor[rgb]{0.1016,0.6523,0.1016}{\ AP\textsubscript{$3$}}}
\psfrag{AP4}[][][0.6]{\textcolor[rgb]{0.1016,0.6523,0.1016}{\ AP\textsubscript{$4$}}}
\psfrag{APm}[][][0.6]{\textcolor[rgb]{0.1016,0.6523,0.1016}{AP\textsubscript{$m$}}}
\psfrag{APM}[][][0.6]{\textcolor[rgb]{0.1016,0.6523,0.1016}{AP\textsubscript{$M$}}}
\psfrag{C1}[][][0.6]{\textcolor[rgb]{0.9961,0.3711,0}{C\textsubscript{$1$}}}
\psfrag{C2}[][][0.6]{\textcolor[rgb]{0.9961,0.3711,0}{C\textsubscript{$2$}}}
\psfrag{C3}[][][0.6]{\textcolor[rgb]{0.9961,0.3711,0}{C\textsubscript{$3$}}}
\psfrag{C4}[][][0.6]{\textcolor[rgb]{0.9961,0.3711,0}{C\textsubscript{$4$}}}
\psfrag{Cm}[][][0.6]{\textcolor[rgb]{0.9961,0.3711,0}{C\textsubscript{$m$}}}
\psfrag{CM}[][][0.6]{\textcolor[rgb]{0.9961,0.3711,0}{\ C\textsubscript{$M$}}}
\psfrag{CU}[][][0.6]{\textcolor[rgb]{0.10,0.10,0.10}{CU}}
\psfrag{Datarecovery}[][][0.6]{\textcolor[rgb]{0.10,0.10,0.10}{\ (Data recovery)}}
\psfrag{User equipment}[][][0.6]{\textcolor[rgb]{0.10,0.5,1.00}{\ \ \ \ \ \ \ \ \ \ \ \ User equipment (UE)}}
\psfrag{Access point}[][][0.6]{\textcolor[rgb]{0.1016,0.6523,0.1016}{\ \ \ \ \ \ \ \ \ \ \ Access point (AP)}}
\psfrag{Backhaul link}[][][0.6]{\textcolor[rgb]{0.9961,0.3711,0}{\ \ \ \ fronthaul link}}
\psfrag{Central unit}[][][0.6]{\textcolor[rgb]{0.10,0.10,0.10}{\ \ \ \ \ \ \ \ \ Central unit (CU)}}
\includegraphics[scale=.3]{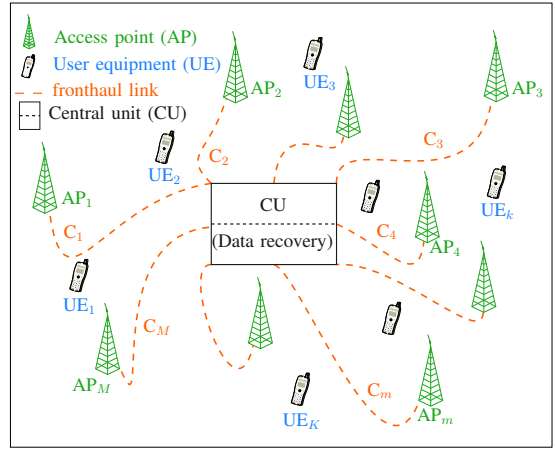}
\caption{{Cell-free mMIMO with limited capacity fronthaul links.}}\label{fig:sysmodel}
\end{figure}

 \subsection{Hardware Impairment Model}
To model the effects of non-ideal hardware at the transmitter/receiver side, it is assumed that the transmitted/received signal becomes distorted with an additive Gaussian noise \cite{bjornson2017massive}. As depicted in Fig.~\ref{fig:HI}, the distorted signal is modeled by
\begin{equation*}
    x_i= \sqrt{\xi_i} x+z_i,
\end{equation*}
where, $x$ is the input signal to the non-ideal hardware, $\xi_{i}\in [{0,1}], \ i = \{t,r\}$ represent the hardware quality factors, where $t$ and $r$ indicate transmitter and receiver, and the distortion caused by the non-ideal hardware is modeled by $z_{i}\!\sim\! \mathcal{CN}\Big(0,(1-\xi_{i})\mathbb{E}\{{\lvert{x}\rvert}^2\}\Big)$ which is independent of $x$. It is worth noting that for the non-ideal hardware model, the input and output signals have the same variance, i.e. $\mathbb{E}\left\lbrace {\lvert x \rvert}^2\right\rbrace = \mathbb{E}\left\lbrace {\lvert x_{i} \rvert}^2\right\rbrace$. Throughout the paper, the terms \emph{perfect} and \emph{useless} hardware respectively indicate $\xi_i=1$ and $\xi_i=0$, and for the sake of simplicity, we assume that all the UEs have the same hardware quality factor $\xi_t$ and all the APs have $\xi_r$.

\subsection{Pilot Transmission}
For the channel training, assume that $\tau$-length orthogonal pilots, i.e. $\boldsymbol{\varphi}_{k}\!\in\!\mathbb{C}^{\tau\times 1}$, are assigned to UEs, where $\tau \!=\! K\! \leq \!T$, and  $\!\boldsymbol{\varphi}_{k}^{H}\boldsymbol{\varphi}_{k^{\prime}}\! =\! \delta(k\!-\!k^{\prime})$ for $k~,~k^{\prime}\in \{1,2,...,K\}$. Therefore, UE\textsubscript{$k$} transmits the pilot signal $\sqrt{\tau\rho_{p}}\boldsymbol{\varphi}_{k}$ and AP\textsubscript{$m$} receives the following signal,
\begin{equation}\label{eqn:Rxpilot}
  \boldsymbol{y}_{p,m} = \sqrt{\xi_{r}}\sum\limits_{k = 1}^{K}g_{mk}\left(\sqrt{\tau\rho_{p}\xi_{t}}\boldsymbol{\varphi}_{k} + \boldsymbol{z}_{t,k}\right) + \boldsymbol{z}_{r,m} + \boldsymbol{n}_{m},
\end{equation}
where, $\rho_{p}$ denotes the pilot power, $\boldsymbol{z}_{t,k}\sim\mathcal{CN}\left(0, \rho_{p}(1-\xi_{t})\boldsymbol{I}_{\tau}\right)$ models the distortion caused by HI at the UE, $\boldsymbol{z}_{r,m}|\{g_{mk}\}\sim\mathcal{CN}\left(0, \rho_{p}(1-\xi_{r})\sum\limits_{k = 1}^{K}{\lvert g_{mk}\rvert}^{2}\boldsymbol{I}_{\tau}\right)$ indicates the HI distortion at the AP\textsubscript{$m$}, and $\boldsymbol{n}_{m}\sim\mathcal{CN}\left(0, N\boldsymbol{I}_{\tau}\right)$ is the additive noise. It is worth mentioning that, $\boldsymbol{n}_{m}$, $\boldsymbol{z}_{t,k}$ and $\boldsymbol{z}_{r,m}$ are assumed to be independent \cite{bjornson2017massive, zhang2017spectral}.
\begin{figure}[!t]
\centering
\psfrag{x}[][][0.75]{$x$}
\psfrag{xtr}[][][0.75]{$\ \ \ x_{t/\color{red}{r}}\! =\! \sqrt{\xi_{t/\color{red}{r}}}x+z_{t/\color{red}{r}}$}
\psfrag{xitr}[][][0.7]{$\sqrt{\xi_{t/\color{red}{r}}}$}
\psfrag{ztr}[][][0.7]{$z_{t/\color{red}{r}}\!\sim\!\mathcal{CN}\Big(\!0,(1-\xi_{t/\color{red}{r}})\mathbb{E}\{{\lvert{x}\rvert}^2\}\!\Big)$}
\includegraphics[scale=.37]{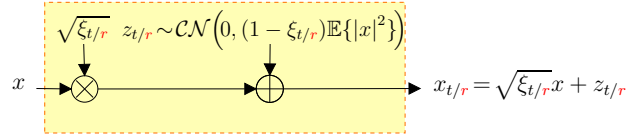}
\caption{{Non-ideal hardware model at the transmitter/receiver.}}\label{fig:HI}
\end{figure}

\subsection{Uplink Data Transmission}
For the data transmission, each UE\textsubscript{$k$} randomly selects information symbol $s_{k} \sim \mathcal{CN}\left(0, 1\right)$ to transmit. Thus, AP\textsubscript{$m$} receives the following superimposed signal,
\begin{equation}\label{eqn:Rxdata}
  {y}_{m} = \sqrt{\xi_{r}}\sum\limits_{k = 1}^{K}g_{mk}\left(\sqrt{\eta_{k}\rho_{u}\xi_{t}}s_{k} + {w}_{t,k}\right) + w_{r,m} + {n}_{m}.
\end{equation}
where, $\rho_{u}$ is the maximum transmit power of each UE and $\eta_{k}\in[0,1]$ is the power control parameter for UE\textsubscript{$k$} . Here, $n_{m}\sim\mathcal{CN}\left(0,N\right)$ denotes the channel additive noise, while ${w}_{t,k}\sim\mathcal{CN}\left(0, \rho_{u}\eta_{k}(1-\xi_{r})\right)$ and ${w}_{r,m}|\{g_{mk}\}\sim\mathcal{CN}\left(0, \rho_{u}(1-\xi_{r})\sum\limits_{k = 1}^{K}\eta_{k}{\lvert g_{mk}\rvert}^{2}\right)$ model the distortion caused by HIs at the UE and AP, respectively.

\subsection{Rate-Distortion Theory}
 To perfectly describe an arbitrary real number, one needs infinite number of bits. Therefore, representing a continuous random variable with \emph{finite} number of bits causes distortion which is precisely analyzed in rate-distortion theory \cite{cover2012elements}. Now the problem is to transmit a random source $x\sim \mathcal{N}(0,P)$ over an error-free fronthaul link with limited dedicated capacity $C$ [bits/s/Hz] which is called the \emph{test channel} as depicted in Fig.~\ref{fig:TestChannel}.  From rate-distortion theory, we need to quantize $x$ to $\hat{x}$ such that the MSE distortion is minimized subject to be able to perfectly transmit the quatized version over the limited capacity link. That is, we have the following optimization problem
\begin{equation}\label{eqn:RateDist}
  R(Q^*)=\underset{Q:~\mathbb{E}\left\lbrace{\lvert \hat{x}-x \rvert}^2\right\rbrace\leq Q}{\min}I(\hat{x};x) \leq C,
\end{equation}
 where $Q$ indicates the distortion caused by quantization and $R(Q^*)$ is the rate-distortion function. Let $\hat{x}=x+q$, where $q\sim \mathcal{N}(0,Q)$ indicates the quantization noise and is independent of $x$, and plug $\hat{x}$ into (\ref{eqn:RateDist}), the minimum achievable value of distortion becomes $Q^*=\frac{P}{2^{2C}-1}$ \cite{cover2012elements}.

\section{Performance Analysis}\label{sec3performance}
After receiving the signals containing pilots, or data signals at the APs, i.e. (\ref{eqn:Rxpilot}) and (\ref{eqn:Rxdata}), each AP can employ the following three strategies for CSI and data signal transmission to the CU;
\begin{itemize}
    \item Compress-forward-estimate (CFE): Each AP compresses the received pilot and data signals separately and forward the compressed versions over the LC-FHL to the CU. Then, channel estimation and data recoveries are carried out at the CU.
    \item Estimate-compress-forward (ECF): First, channel estimation is performed at each AP, and then AP separately compresses the estimated channels and data signals, and forwards the compressed signals to the CU. Finally, the CU recovers CSI and performs data detection.
    \item Estimate-multiply-compress-forward (EMCF): Each AP first estimates the channels, then multiplies the received data signal by the conjugate of the estimated channels, and compresses and forwards the results to the CU. Thus, CU only performs data detection.
\end{itemize}
Since CFE strategy has the lowest complexity at the AP, it is well suited from the APs points of view to have low processing units. In contrast, EMCF one has the highest complexity due to the channel estimation and multiplication at the APs. On the other hand, complexity of ECF strategy is in between the two other methods. In the following subsections, achievable sum-rates of these strategies are derived and discussed.

\begin{figure}[!t]
\centering
\psfrag{xrt}[][][0.85]{$\ \ \ \ \ \ \ x$}
\psfrag{Q}[][][0.8]{$Q(.)$}
\psfrag{yq}[][][0.85]{$ \ \ \ \ \ \ \ \hat{x}\!=\!x\!+\!q\!$}
\psfrag{fronthaullink}[][][0.85]{\ \ \small Error-free LC-FHL}
\psfrag{C}[][][0.8]{$C$ [bits/s/Hz]}
\psfrag{xhat}[][][0.85]{$\hat{x}$}
\includegraphics[scale=.37]{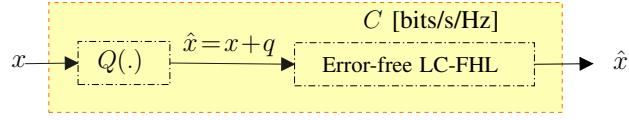}
\caption{{Rate-distortion theoretic test channel.}}\label{fig:TestChannel}
\vspace{-0.5 cm}
\end{figure}

\subsection{Compress-Forward-Estimate Strategy}
In the following, we analyze the CFE strategy in presence of HIs and LC-FHL. First, channel estimation is presented and afterwards the achievable rate is derived.
\subsubsection{Channel Estimation}
After receiving the pilots signal $\boldsymbol{{y}}_{p,m}$ given in (\ref{eqn:Rxpilot}), AP\textsubscript{$m$} compresses this vector element wise. Therefore, we have
\begin{equation}\label{eqn:QpilotCFE}
  \boldsymbol{\hat{y}}_{p,m} = \boldsymbol{{y}}_{p,m} + \boldsymbol{{q}}_{p,m},
\end{equation}
where, $\boldsymbol{q}_{p,m}$ is the quantization noise vector where its elements are i.i.d according to $\mathcal{CN}\left(0,Q_{p,mk}\right)$. Then, the AP forwards $\boldsymbol{\hat{y}}_{p,m}$ to the CU to estimate the channel via applying linear MMSE (LMMSE) estimator. So, we have
\begin{equation}\label{eqn:est_CFE1}
  \tilde{g}_{mk} = \frac{\mathbb{E}\{\bar{y}_{p,mk}g_{mk}^{*}\}}{\mathbb{E}\{{\lvert \bar{y}_{p,mk}\rvert}^2\}}\bar{y}_{p,mk}=:\lambda_{mk}\bar{y}_{p,mk},
\end{equation}where
\begin{align}\label{eqn:CFEMMSE}
  \bar{y}_{p,mk}\!\! &=\!\! \boldsymbol{\varphi}_{k}^{H}\boldsymbol{\hat{y}}_{p,m}\! =\!  \sqrt{\xi_{r}\xi_{t}\tau\rho_{p}}g_{mk} + \boldsymbol{\varphi}_{k}^{H}\sum\limits_{k^{\prime}=1}^{K}\sqrt{\xi_{r}}g_{mk^{\prime}}\boldsymbol{z}_{t,k^{\prime}} + \boldsymbol{\varphi}_{k}^{H}\boldsymbol{z}_{r,m} + \boldsymbol{\varphi}_{k}^{H}\boldsymbol{n}_{m} + \boldsymbol{\varphi}_{k}^{H}\boldsymbol{q}_{p,m},\\\label{eqn:CFElambda}
  \lambda_{mk}\!\! &=\!\! \frac{\sqrt{\xi_{r}\xi_{t}\tau\rho_{p}}\beta_{mk}}{\xi_{r}\xi_{t}\tau\rho_{p}\beta_{mk} + \rho_{p}(1-\xi_{r}\xi_{t})\sum\limits_{k^{\prime}=1}^{K}\beta_{mk^{\prime}} + N + \frac{1}{\tau}\sum\limits_{k^{\prime}=1}^{K}Q_{p,mk^{\prime}}}.
\end{align}
Variance of the channel estimation given in (\ref{eqn:est_CFE1}) is $\gamma_{mk}:=\mathbb{E}\{{\lvert \tilde{g}_{mk} \rvert}^{2}\}=\sqrt{\xi_{r}\xi_{t}\tau\rho_{p}}\beta_{mk}\lambda_{mk}$. It is worth noting that, due to the product of Gaussian random variables in $\bar{y}_{p,mk}$, it does not have Gaussian distribution. As a result, LMMSE is not an optimal estimator but it still results in a good approximation \cite{bjornson2017massive} and estimation error and estimated variable are only uncorrelated.
\subsubsection{Uplink Achievable Rate}
After receiving the data signal $y_m$ given in (\ref{eqn:Rxdata}), AP\textsubscript{$m$} compresses this signal as follows
\begin{equation}\label{eqn:QdataCFE}
  \hat{y}_{m} = {y}_{m} + {q}_{d,m},
\end{equation}
where $q_{d,m}\sim\mathcal{CN}\left(0,Q_{d,m}\right)$ is the quantization noise. Then, the AP forwards $\hat{y}_{m}$ to the CU over the LC-FHL.

After receiving all the quantized symbols, to recover each user's data, the CU performs MRC technique based on the estimated channels given in (\ref{eqn:est_CFE1}). Therefore, to recover data of UE\textsubscript{$k$} by utilizing the well-known MRC along with  use-and-then-Forget (UatF) approach \cite{marzetta2016fundamentals}, the effective received symbol for UE\textsubscript{k} becomes
\begin{equation*}
  r_{k} = \sum\limits_{m=1}^{M}\hat{y}_{m}\tilde{g}_{mk}^{*}\! = \underbrace{\mathbb{E}\!\left\lbrace\!\sum\limits_{m=1}^{M}\sqrt{\rho_{u}\eta_{k}\xi_{r}\xi_{t}}g_{mk}\tilde{g}_{mk}^{*}\!\!\right\rbrace}_{\text{DS}_k}\!s_{k} \!+\!\! \underbrace{\sqrt{\rho_{u}\eta_{k}\xi_{r}\xi_{t}}\!\!\left\lbrace\!\sum\limits_{m=1}^{M}g_{mk}\tilde{g}_{mk}^{*} \!-\! \mathbb{E}\!\left\lbrace\!\sum\limits_{m=1}^{M}g_{mk}\tilde{g}_{mk}^{*}\!\!\right\rbrace\!\!\right\rbrace}_{\text{BU}_k}\!s_{k}\\
\end{equation*}
\begin{equation}\label{eqn:MRC_CFE}
  +\!\! \sum\limits_{k^{\prime}\neq k}^{K}\underbrace{\!\!\sqrt{\rho_{u}\eta_{k^{\prime}}\xi_{r}\xi_{t}}\!\sum\limits_{m=1}^{M}g_{mk^{\prime}}\tilde{g}_{mk}^{*}}_{\text{IUI}_{kk^{\prime}}}\!s_{k^{\prime}} +\!\!\sum\limits_{k^{\prime}= 1}^{K}\!\underbrace{\!\sqrt{\xi_{r}}\!\sum\limits_{m=1}^{M}\!\!g_{mk^{\prime}}\tilde{g}_{mk}^{*}w_{t,k^{\prime}}}_{\text{THI}_{kk^{\prime}}}\! + \underbrace{\!\!\sum\limits_{m=1}^{M}\!\!w_{r,m}\tilde{g}_{mk}^{*}}_{\text{RHI}_k}\! + \underbrace{\!\!\sum\limits_{m=1}^{M}\!\!n_{m}\tilde{g}_{mk}^{*}}_{\text{RN}_k}\!+ \underbrace{\!\!\sum\limits_{m=1}^{M}\!\!q_{d,m}\tilde{g}_{mk}^{*}}_{\text{QN}_k}\!,
\end{equation}
where,
\begin{itemize}
  \item DS\textsubscript{$k$} indicates the desired part of the signal,
  \item BU\textsubscript{$k$} is the beamforming uncertainty due to statistical knowledge of the channels only,
  \item IUI\textsubscript{$kk^{\prime}$} denotes the inter-user interference between UE\textsubscript{$k$} and UE\textsubscript{$k^{\prime}$},
  \item THI\textsubscript{$kk^{\prime}$} represents interference as a result of HI at the transmitters, i.e. UEs,
  \item RHI\textsubscript{$k$} is interference due to HI at the receivers, i.e. APs,
  \item RN\textsubscript{$k$} and QN\textsubscript{$k$} are receiver and compression noises, respectively.
\end{itemize}
These terms are pair-wisely uncorrelated, and assuming them as an equivalent Gaussian random variables, to analyze the worst case scenario, the following uplink achievable data rate for UE\textsubscript{$k$} is obtained.
\begin{equation}\label{GeneralRate}
\begin{aligned}
    \!R_{k} \!\!&= \!\!\frac{T\!-\!\tau}{T}\!\log_{2}\!\!\left(\!1 \!+ \!\text{SINR}_{k}^{CFE}\right)\!\!,\\
    \!\text{SINR}_{k}^{CFE} \!\!&= \!\! \frac{{\!\lvert\text{DS}_{k}\rvert}^2\!}{\mathbb{E}\!\!\left\lbrace{\!\lvert\text{BU}_{k}\rvert}^2\!\right\rbrace \!+\!\!\!\sum\limits_{k^{\prime}\neq k}^{K}\!\!\mathbb{E}\!\!\left\lbrace{\!\lvert\text{IUI}_{kk^{\prime}}\rvert}^2\!\right\rbrace \!+\! \!\!\sum\limits_{k^{\prime} = 1}^{K}\!\!\mathbb{E}\!\!\left\lbrace{\!\lvert\text{THI}_{kk^{\prime}}\rvert}^2\!\right\rbrace \!+\! \mathbb{E}\!\!\left\lbrace{\!\lvert\text{RHI}_{k}\rvert}^2\!\right\rbrace \!+\! \mathbb{E}\!\!\left\lbrace{\!\lvert\text{RN}_{k}\rvert}^2\!\right\rbrace \!+\! \mathbb{E}\!\!\left\lbrace{\!\lvert\text{QN}_{k}\rvert}^2\!\right\rbrace}.
\end{aligned}
\end{equation}
\begin{theorem}
  For the CFE strategy, UE\textsubscript{$k$} has the following signal to noise and interference ratio
\begin{equation}\label{eqn:RateCFE1}
  \text{SINR}_{k}^{CFE}\!= \!\!\frac{\!\rho_{u}\eta_{k}\xi_{r}\xi_{t}\Gamma_{kk} \!}{\!\!\sum\limits_{k^{\prime}=1}^{K}\!\!\!\rho_{u}\eta_{k^{\prime}}\!\!\Bigg[\!\Omega_{kk^{\prime}}\!\!+\!\! \xi_{r}(\!1\!\!-\!\xi_{t}\!)\!\!\left(\!\!\boldsymbol{\varphi}_{k}^{H}\!\boldsymbol{\varphi}_{k^{\prime}}\!\!-\!\frac{1}{\!\tau\xi_{t}\!}\!\right)\!\Gamma_{kk^{\prime}} \!\! + \!\!\rho_{p}(\!1\!\!-\!\xi_{r}\!)\!\Big(\!\tau\xi_{r}\xi_{t}\boldsymbol{\varphi}_{k}^{H}\!\boldsymbol{\varphi}_{k^{\prime}}\!\!-\!(\!1\!+\!\xi_{r}\!\!-\!\xi_{r}\xi_{t}\!)\!\Big)\!\Lambda_{kk^{\prime}}\!\!\Bigg]\!\!+\!\!\mathcal{E}_{k}},
\end{equation}
\begin{equation*}
 \Gamma_{kk^{\prime}} \!=\!\!\! \left(\sum\limits_{m=1}^{M}\!\!\gamma_{mk}\frac{\beta_{mk^{\prime}}}{\beta_{mk}}\!\right)^2, \ \ \ \ \ \Omega_{kk^{\prime}} \!= \!\!\!\sum\limits_{m=1}^{M}\!\!\gamma_{mk}\beta_{mk^{\prime}}, \ \ \ \ \ \Lambda_{kk^{\prime}} \!=\! \!\sum\limits_{m=1}^{M}\!\!\lambda_{mk}^{2}\beta_{mk^{\prime}}^{2} \ \ \ \ \ \mathcal{E}_{k}\!=\!\!\!\sum\limits_{m=1}^{M}\!\!(\!N\!\!+\!Q_{d,m})\!\gamma_{mk}.
\end{equation*}
\end{theorem}
\begin{proof}
  The proof is given in Appendix A.
\end{proof}
\begin{Remark}
  For infinite-capacity FHLs, i.e. $C_m=\infty$ or equivalently $Q_{d,m} = Q_{p,mk} = 0$ for all $m$ and $k$, equation (\ref{eqn:RateCFE1}) reduces to \cite[(17)  with orthogonal pilots]{zhang2018performance}. Also, if hardwares are perfect, the results reduces to \cite[equation (27) with orthogonal pilots]{ngo2017cell}.
 \end{Remark}

Now, based on the rate-distortion theory, we derive the variance of quatization noises given in (\ref{eqn:QdataCFE}) and (\ref{eqn:QpilotCFE}) as functions of $C_m$. Let $C_{m}=C_{p,m}+C_{d,m}$, where $C_{p,m}$ denotes part of the fronthaul capacity dedicated for forwarding quantized pilot vector $\boldsymbol{\hat{y}}_{p,m}$ and $C_{d,m}$ is the dedicated capacity for forwarding $y_m$.
\begin{itemize}
    \item For computing $C_{p,m}$, we have
    \begin{equation}\label{PBHCFE}
    \begin{split}
    \hspace{-0.8cm} C_{p,m} &= \frac{1}{T}I(\boldsymbol{y}_{p,m};\boldsymbol{\hat{y}}_{p,m}) = \frac{1}{T}\Big[h(\boldsymbol{\hat{y}}_{p,m})-h(\boldsymbol{\hat{y}}_{p,m}|\boldsymbol{y}_{p,m})\Big] = \frac{1}{T}\Big[h(\boldsymbol{y}_{p,m}+\boldsymbol{q}_{p,m})-h(\boldsymbol{q}_{p,m})\Big]\\
    &\overset{\text{(a)}}\leq \frac{1}{T}\Big[\!\log\left((2\pi e)^\tau\lvert{\mathbb{E}\left\lbrace(\boldsymbol{y}_{p,m}\!+\!\boldsymbol{q}_{p,m})(\boldsymbol{y}_{p,m}\!+\!\boldsymbol{q}_{p,m})^H\right\rbrace}\rvert\right) \!-\! \log\left((2\pi e)^\tau\lvert{\mathbb{E}\left\lbrace\boldsymbol{q}_{p,m}\boldsymbol{q}_{p,m}^H\right\rbrace}\rvert\right)\!\Big]\\
    &\overset{\text{(b)}}= \frac{K}{T}\log\Bigg[1+\frac{\rho_{p}\sum\limits_{k=1}^{K}\beta_{mk}+N}{Q_{p,m}}\Bigg], \ \ \
    \end{split}
    \end{equation}
    where, (a) follows from the \emph{maximum differential entropy lemma} in \cite[Chapter 2]{el2011network}. Since elements of the vector $\boldsymbol{y}_{p,m}$ have the same variance, they are quantized in the same manner. Thus, we have $Q_{p,m}=Q_{p,m1}=...=Q_{p,mK}$, and as a result (b) holds for $\tau=K$.
    \item For computing $C_{d,m}$, we have
    \begin{equation}\label{DBHCFE}
    \hspace{-0cm} C_{d,m}\! =\! \frac{\!T\!-\!\tau\!}{T}I({y}_{m};{\hat{y}}_{m}) \!\leq\! \frac{\!T\!-\!\tau\!}{T}\log\!\Bigg[\!\frac{\mathbb{E}\left\lbrace{\lvert y_{m}\rvert}^2\right\rbrace \!+\! Q_{d,m}}{Q_{d,m}}\!\Bigg]\!=\! \frac{\!T\!-\!\tau\!}{T}\log\!\Bigg[\!1\!+\!\frac{\rho_{u}\sum\limits_{k=1}^{K}\eta_{k}\beta_{mk}\!+\!N}{Q_{d,m}}\!\Bigg].
    \end{equation}
\end{itemize}
Therefore, from (\ref{PBHCFE}) and (\ref{DBHCFE}), one can find variances of the quantization noises, i.e. $Q_{p,m}$ and $Q_{d,m}$, such that  $C_{p,m}+C_{d,m} = C_{m}$.

\subsection{Estimate-Compress-Forward Strategy}
Here, we analyze the ECF strategy; First, the LMMSE channel estimation is preformed at the AP, then AP separately quatizes the channel coefficients and the received superimposed data signal, and forward them over the LC-FHL. Finally, lower and upper bonds on the achievable rates are presented.
\subsubsection{Channel Estimation}
Since the AP\textsubscript{$m$} first estimates the $K$ channels corresponding to $K$ UEs, by plugging $Q_{p,mk}=0$ and replacing $\boldsymbol{\hat{y}}_{p,m}$ with $\boldsymbol{{y}}_{p,m}$ in (\ref{eqn:QpilotCFE})--(\ref{eqn:CFElambda}), we have
\begin{equation}\label{eqn:est_ECF1}
  \tilde{g}_{mk} = \lambda_{mk}\bar{y}_{p,mk} = \frac{\mathbb{E}\{\bar{y}_{p,mk}g_{mk}^{*}\}}{\mathbb{E}\{{\lvert \bar{y}_{p,mk}\rvert}^2\}}\bar{y}_{p,mk},
\end{equation}
where
\begin{equation}\label{eqn:est_ECF2}
  \bar{y}_{p,mk}\! =\! \boldsymbol{\varphi}_{k}^{H}\boldsymbol{{y}}_{p,m}, \ \ \ \ \ \ \ \ \ \ \lambda_{mk}\! =\! \frac{\sqrt{\xi_{r}\xi_{t}\tau\rho_{p}}\beta_{mk}}{\xi_{r}\xi_{t}\tau\rho_{p}\beta_{mk} + \rho_{p}(1-\xi_{r}\xi_{t})\sum\limits_{k^{\prime}=1}^{K}\beta_{mk^{\prime}} + N}.
\end{equation}
Thus, the variance of the estimated channel is $\mathbb{E}\!\!\left\lbrace{\!\lvert\tilde{g}_{mk}\rvert}^2\!\right\rbrace\! =\! \gamma_{mk}\! =\!\! \sqrt{\xi_{r}\xi_{t}\tau\rho_{p}}\beta_{mk}\lambda_{mk}$.

Then, AP\textsubscript{$m$} quantizes each estimated channel, and then forward them over the LC-FHL. Thus, CU receives $\hat{g}_{mk}, \ \forall \{m, k\}$ through the following test channel, as follows
\begin{equation}\label{PtestECF}
  \tilde{g}_{mk} = \hat{g}_{mk} + q_{p,mk}, \forall \{m, k\}.
\end{equation}
One can show that $\gamma_{mk}^{\prime}:=\mathbb{E}\left\lbrace{\lvert\hat{g}_{mk}\rvert}^2\right\rbrace = \mathbb{E}\left\lbrace\hat{g}_{mk}^{*}g_{mk}\right\rbrace = \mathbb{E}\left\lbrace{\lvert\tilde{g}_{mk}\rvert}^2\right\rbrace - Q_{p,mk}$.
\subsubsection{Uplink Achievable Rate}
For ECF strategy, quantizing the received data signal and forwarding it to the CU is the same as that of CFE one. The only difference is that the CU uses $\hat{g}_{mk}^{*}$, instead of $\tilde{g}_{mk}^{*}$ in (\ref{eqn:MRC_CFE}), for applying MRC and UatF techniques to derive the effective signal $r_{k}$ for UE\textsubscript{$k$}. Besides, in contrast to the CFE, due to the correlation between the received signals of different APs which is stemmed from the hardware impairments, computing achievable rate becomes challenging. In the following Theorems, we present a lower and a upper bounds for the achievable rates.
\begin{theorem}
Using equation (\ref{GeneralRate}), lower bound of SINR of UE\textsubscript{$k$} is given by
\begin{equation}\label{eqn:RateCFE1}
  \text{SINR}_{k,LB}^{ECF}\!= \!\frac{\!\rho_{u}\eta_{k}\xi_{r}\xi_{t}\left(\sum\limits_{m=1}^{M}\gamma_{mk}^{\prime}\right)^2 \!}{\!\!\sum\limits_{k^{\prime}=1}^{K}\!\!\rho_{u}\eta_{k^{\prime}}\!\!\Bigg[\!\Omega_{kk^{\prime}}^{\prime}\!\!+\!\! \xi_{r}(\!1\!\!-\!\!\xi_{t}\!)\!\left(\!\boldsymbol{\varphi}_{k}^{H}\!\boldsymbol{\varphi}_{k^{\prime}}\!\!-\!\frac{1}{\!\tau\xi_{t}\!}\!\!\right)\!\!\Gamma_{kk^{\prime}} \!\! +\! \!\rho_{p}(\!1\!\!-\!\!\xi_{r}\!)\!\Big(\!\tau\xi_{r}\xi_{t}\boldsymbol{\varphi}_{k}^{H}\!\boldsymbol{\varphi}_{k^{\prime}}\!-\!(\!1\!+\!\xi_{r}\!-\!\xi_{r}\xi_{t}\!)\!\Big)\!\Lambda_{kk^{\prime}}\!\!\Bigg]\!+\!\mathcal{E}_{k}^{\prime}},
\end{equation}
\vspace{-0.7cm}
\begin{equation*}
\begin{split}
\Omega_{kk^{\prime}}^{\prime} =& \sum\limits_{m=1}^{M}\gamma_{mk}^{\prime}\beta_{mk^{\prime}},\ \ \ \ \ \  \mathcal{E}_{k}^{\prime}=\sum\limits_{m=1}^{M}\!\!(\!N\!\!+\!Q_{d,m})\!\gamma_{mk}^{\prime}\!- \!\!\! \sum\limits_{k^{\prime}\neq k}^{K}\!\!\!\rho_{u}\eta_{k^{\prime}}\!\Bigg[\!\Big(\rho_{p}(1-\xi_{r})^{2}\!+\!\xi_{r}\Big)\!\!\sum\limits_{m=1}^{M}\!\!Q_{p,mk}Q_{p,mk^{\prime}}\!\Bigg]\!\\
-&\rho_{u}\eta_{k}\!\Bigg[\!\Big(\xi_{r}(1-\xi_{t})\!+\!\rho_{p}(1-\xi_{r})^2\!\Big)\!\sum\limits_{m=1}^{M}\!\!Q_{p,mk}^{2}\!-\!2\xi_{r}\xi_{t}\Big(\sum\limits_{m=1}^{M}Q_{p,mk}\Big)\Big(\sum\limits_{m=1}^{M}\gamma_{mk}^{\prime}\Big)\Bigg].
\end{split}
\end{equation*}
\end{theorem}
\begin{proof}
  The proof is given in Appendix B.
\end{proof}
\begin{theorem}
Using equation (\ref{GeneralRate}), upper bound of SINR\textsubscript{$k$} is given by
\begin{equation}\label{eqn:RateCFE1}
  \text{SINR}_{k,UB}^{ECF}\!= \!\frac{\!\rho_{u}\eta_{k}\xi_{r}\xi_{t}\left(\sum\limits_{m=1}^{M}\gamma_{mk}^{\prime}\right)^2 \!}{\!\!\sum\limits_{k^{\prime}=1}^{K}\!\!\rho_{u}\eta_{k^{\prime}}\!\!\Bigg[\!\!\sum\limits_{m=1}^{M}\!\!\gamma_{mk}^{\prime}\beta_{mk^{\prime}}\!\Bigg]\!\!+\! \rho_{u}\eta_{k}\!\Bigg[\!\xi_{r}\!(\!1\!-\!\xi_{t})\!\Big(\!\sum\limits_{m=1}^{M}\!\!\!\gamma_{mk}^{\prime}\!\Big)^{2}\!\!+\!(1\!\!-\!\xi_{r})\!\!\!\sum\limits_{m=1}^{M}\!\!\!\gamma_{mk}^{{\prime}^2}\!\Bigg] \!\!+\!\!\sum\limits_{m=1}^{M}\!\!(N\!+\!Q_{d,m}\!)\gamma_{mk}^{\prime}}\!.
\end{equation}
\end{theorem}
\begin{proof}[Sketch of Proof]
By assuming that the estimation error $e_{mk}\sim\mathcal{CN}(0,\beta_{mk}-\gamma_{mk})$ is independent of the estimated channel, one can compute $\mathbb{E}\!\left\lbrace{\!\lvert\text{BU}_{k}\rvert}^2\!\right\rbrace$, $\mathbb{E}\!\!\left\lbrace{\!\lvert\text{IUI}_{kk^{\prime}}\rvert}^2\!\right\rbrace$, $\mathbb{E}\!\!\left\lbrace{\!\lvert\text{THI}_{kk^{\prime}}\rvert}^2\!\right\rbrace$, and $\mathbb{E}\!\!\left\lbrace{\!\lvert\text{RHI}_{k}\rvert}^2\!\right\rbrace$  by replacing $g_{mk} = \underbrace{\hat{g}_{mk} + q_{p,mk}}_{\tilde{g}_{mk}} + e_{mk}$. Hence, because of independency assumption, some \emph{non-negative} expected values become zero. Thus, a lower-bound for the variance of interference terms are obtained and as a results, we have an upper bound for the achievable rate.
\end{proof}
\begin{Remark}
 For perfect hardware units at the UEs and APs, and infinite fronthaul capacity, the upper and lower bounds are tight and tend to the result obtained in \cite[equation (27) with orthogonal pilots]{ngo2017cell}.
\end{Remark}

Now, based on the rate-distortion theory, we need to derive the variance of quatization noises as functions of $C_m$. For forwarding the quatized data signal over the LC-FHL, the same result as CFE holds, see (\ref{DBHCFE}). For qunatizing CSI and forwarding them over test channel (\ref{PtestECF}), we have
\begin{equation}\label{PBHECF}
C_{p,m}\! =\! \frac{1}{T}I(\boldsymbol{\tilde{g}}_{m};\boldsymbol{\hat{g}}_{m})\! = \! \frac{1}{T}\!\Big[\!h(\boldsymbol{\tilde{g}}_{m})\!-\!h(\boldsymbol{\tilde{g}}_{m}|\boldsymbol{\hat{g}}_{m})\!\Big] \!\leq\! \frac{1}{T}\!\log\!\!\Bigg[\frac{\lvert{\mathbb{E}\left\lbrace\boldsymbol{\tilde{g}}_{m}\boldsymbol{\tilde{g}}_{m}^H\right\rbrace}\rvert}{\lvert{\mathbb{E}\left\lbrace\boldsymbol{q}_{p,m}\boldsymbol{q}_{p,m}^H\right\rbrace}\rvert}\Bigg] =\! \frac{1}{T}\!\!\sum\limits_{k=1}^{K}\!\!\log\!\!\Bigg[\frac{\gamma_{mk}}{Q_{p,mk}}\Bigg].
\end{equation}
where, $\boldsymbol{\tilde{g}_{m}} := [\tilde{g}_{m1},\tilde{g}_{m2},...,\tilde{g}_{mK}]^{T}$ and $\boldsymbol{\hat{g}_{m}} := [\hat{g}_{m1},\hat{g}_{m2},...,\hat{g}_{mK}]^{T}$.
\begin{Remark}
It is worth mentioning that finding optimal values of $Q_{p,mk}$ to maximize the achievable sum-rate is a challenging task, due to non-convexity of the achievable rate as a function of these variables. Here, exhaustive search due to the large number of APs and UEs is prohibitively time demanding. Thus, we propose a water-filling based low-complexity approach in which the UEs with better channels are allocated more fronthaul capacity for CSI transmission. For the proposed scheme, we offer to find $Q_{p,mk}$ in (\ref{PBHECF}) such that the following equality satisfies.
\begin{equation}\label{LCEFC}
  \log\Big[\frac{\gamma_{mk}}{Q_{p,mk}}\Big] = \frac{\gamma_{mk}}{\sum_{k=1}^{K}\gamma_{mk}}TC_{p,m}.
\end{equation}
Moreover, optimal values of $C_{p,m}$ and $C_{d,m}$ can be found by a simple one-dimensional line search to have $C_{p,m}+C_{d,m}=C_m$. Noting that the proposed approaches are only rely on large-scale fading parameters which vary slowly hence, can be carried out at the CU with a slight load on the fronthaul link. Through the numerical results, it is shown that for very limited-capacity FHLs, this approach boosts the performance of the system substantially.
\end{Remark}

\begin{proposition}
  At high SNR regime, if $\sum\limits_{k^{\prime}=1}^{K}\!\!\gamma_{mk^{\prime}}^{\infty}\!>\!K\gamma_{mk}^{\infty}$ and $C_{p,m}\!>\!C_{m,th}$, where $C_{p,m}$ is the fronthaul capacity allocated for CSI transmission, estimating the channels at the CU results in smaller estimation error.
 \begin{equation}\label{eqn:remark3}
  C_{m,th}\!=\!\Big[\!\frac{T}{K}-\theta_{2}T\!\Big]^{\!-1}\!\!\!\!\log\!\!\left(\!\!\theta_{1}^{-1}\sum\limits_{k^{\prime}=1}^{K}\beta_{mk^{\prime}}\!\!\right)\!,
\end{equation}
where
\begin{equation*}
    \theta_{1} = \xi_{r}\xi_{t}\tau\beta_{mk}\!+\!(1\!-\!\xi_{r}\xi_{t})\!\!\!\sum\limits_{k^{\prime}=1}^{K}\!\!\beta_{mk^{\prime}}, \ \ \ \theta_{2} = \frac{\gamma_{mk}^{\infty}}{\sum_{k=1}^{K}\!\gamma_{mk}^{\infty}}, \ \ \ \gamma_{mk}^{\infty} = \theta_{1}^{-1}\xi_{r}\xi_{t}\tau\beta_{mk}^{2}\!.
\end{equation*}
\end{proposition}

\begin{proof}
  At high SNR regime, from equations, (\ref{eqn:est_CFE1}), (\ref{PBHCFE}), (\ref{eqn:est_ECF1}), (\ref{PtestECF}), and (\ref{PBHECF}), one can obtain the variance of the estimated channels for the two schemes as follows,
  \begin{equation*}
    \begin{split}
       \gamma_{mk}^{\text{CFE}} &= \lim_{\rho_{p}\rightarrow\infty} \gamma_{mk}\overset{(\ref{eqn:est_CFE1}, \ref{PBHCFE})}{=} \frac{\xi_{r}\xi_{t}\tau\beta_{mk}^{2}}{\xi_{r}\xi_{t}\tau\beta_{mk} + \Big[(1-\xi_{r}\xi_{t}) + \Big[2^{\frac{T}{K}C_{p,m}}-1\Big]^{-1}\Big]\sum\limits_{k^{\prime}=1}^{K}\beta_{mk^{\prime}}}\\
       \gamma_{mk}^{\text{ECF}} &= \lim_{\rho_{p}\rightarrow\infty} \gamma_{mk}^{\prime}\overset{(\ref{eqn:est_ECF1}, \ref{PtestECF}, \ref{PBHECF})}{=}\gamma_{mk}^{\infty} - \gamma_{mk}^{\infty}\times 2^{-\gamma_{mk}^{\infty}TC_{p,m}\Big[{\sum\limits_{k^{\prime}=1}^{K}\gamma_{mk^{\prime}}^{\infty}}\Big]^{-1}}.
    \end{split}
  \end{equation*}
  For CFE, to have smaller estimation error at the CU following inequality should hold.
  \begin{equation*}
    \beta_{mk}-\gamma_{mk}^{\text{CFE}} < \beta_{mk}-\gamma_{mk}^{\text{ECF}}.
  \end{equation*}
  After mathematical manipulations, above inequality reduces to, $\left(2^{C_{p,m}T/K}-1\right)\left(2^{C_{p,m}\theta_{2}T}-1\right)^{-1}>\theta_{1}^{-1}\sum\limits_{k^{\prime}=1}^{K}\beta_{mk^{\prime}}$; By further assuming that $1\!>\!K\theta_{2}$ and $C_{p,m}$ is large enough, this inequality can be approximated as  $\left(2^{C_{p,m}(T/K-\theta_{2}T)}\right)>\theta_{1}^{-1}\sum\limits_{k^{\prime}=1}^{K}\beta_{mk^{\prime}}$. from the last inequality $C_{m,th}$ is obtained.
\end{proof}

\begin{proposition}
   For a single-user CF-mMIMO scenario, $K = 1$, and high SNR regime, the ECF strategy outperforms the CFE one.
\end{proposition}
\begin{proof}
  For $K=\tau = 1$, from Theorems 1 and 3, and equations (\ref{eqn:est_CFE1}), (\ref{PBHCFE}), (\ref{eqn:est_ECF1}), (\ref{PtestECF}), and (\ref{PBHECF}), it is proven that,
  \begin{equation*}
  \begin{split}
       \Upsilon_{m} &= \lim_{\rho_{p}\rightarrow\infty} \gamma_{mk}\overset{(\ref{eqn:est_CFE1}, \ref{PBHCFE})}{=} \lim_{\rho_{p}\rightarrow\infty} \gamma_{mk}^{\prime}\overset{(\ref{eqn:est_ECF1}, \ref{PtestECF}, \ref{PBHECF})}{=}\xi_{r}\xi_{t}\left(1-2^{-TC_{p,m}}\right)\beta_{m},\\
       \text{SINR}_{\text{CFE}}^{\infty} &=\lim_{(\rho_{u},\rho_{p})\rightarrow\infty}\text{SINR}_{\text{CFE}}\overset{\text{Th. }1}{=} \frac{\xi_{r}\xi_{t}\mathcal{X}_{0}}{\mathcal{X}_{1}+a\mathcal{X}_{0}+b\mathcal{X}_{2}+\mathcal{X}_{3}},\\
       \text{SINR}_{\text{ECF}}^{\infty} &=\lim_{(\rho_{u},\rho_{p})\rightarrow\infty}\text{SINR}_{\text{ECF}}\overset{\text{Th. }3}{=} \frac{\xi_{r}\xi_{t}\mathcal{X}_{0}}{\mathcal{X}_{1}+\xi_{r}(1-\xi_{r})\mathcal{X}_{0}+(1-\xi_{r})\mathcal{X}_{2}+\mathcal{X}_{3}},\\
       \mathcal{X}_{0} &=\! \left(\!\sum\limits_{m=1}^{M}\!\Upsilon_{m}\!\right)^{2}\!\!, \ \ \mathcal{X}_{1}\! =\! \sum\limits_{m=1}^{M}\!\Upsilon_{m}\beta_{m}, \ \ \mathcal{X}_{2}\! =\! \sum\limits_{m=1}^{M}\!\Upsilon_{m}^{2}, \ \ \mathcal{X}_{3}\! =\! \sum\limits_{m=1}^{M}\!\Upsilon_{m}\beta_{m}\left(\!2^{\frac{T}{T-1}C_{d,m}}-1\!\right)^{-1}\!\!\!,\\
       a &=\xi_{r}\left(1-\xi_{r}+\frac{1-\xi_{r}}{\xi_{r}}\right), \ \ b=(1-\xi_{r})\left(1+\frac{1}{\xi_{r}\xi_{t}}+\frac{1-\xi_{r}}{\xi_{t}}\right).
       \end{split}
  \end{equation*}
  Since $0\leq \xi_{r},\xi_{t}\leq 1$, $a\geq \xi_{r}(1-\xi_{r})$ and $b\geq (1-\xi_{r})$ we have $\text{SINR}_{\text{ECF}}^{\infty}\geq \text{SINR}_{\text{CFE}}^{\infty}$.
\end{proof}
\subsection{Estimate-Multiply-Compress-Forward Strategy}
In the following subsections, we present the processing at APs and CU.
\subsubsection{Processing at APs}
AP\textsubscript{$m$} first estimates the UEs' channels via LMMSE estimator, then it separately multiplies the conjugate of the estimated channels to its received signal and forms the following vector.
\begin{equation}\label{eqn:mult_vec}
    \boldsymbol{\tilde{y}}_{m}\! =\!\! [\tilde{g}_{m1}^{*}y_{m}, \ \tilde{g}_{m2}^{*}y_{m}, \ \tilde{g}_{mK}^{*}y_{m}]^T,
\end{equation}
where, $\tilde{g}_{mk}$ denotes the estimated channel of UE\textsubscript{k} given in (\ref{eqn:est_ECF1})--(\ref{eqn:est_ECF2}), and $y_m$ is the received superimposed data signal given in (\ref{eqn:Rxdata}). Afterwards, the AP quatizes $\boldsymbol{\tilde{y}}_{m}$ as follows
\begin{equation*}
    \boldsymbol{\hat{y}}_{m}=\boldsymbol{\tilde{y}}_{m}+\boldsymbol{q}_m,
\end{equation*}
where $\boldsymbol{q}_m$ denotes the quantization noise vector, where $q_{mk}\sim\mathcal{CN}(0,Q_{mk})$ for $k=\{1,2,...,K\}$. To perfectly deliver the quantized vector to the CU over the FHL with capacity $C_m$, we must have
\begin{equation}\label{BHECMF2}
 \begin{split}
  C_{m} &= \frac{T-\tau}{T}I\left(\boldsymbol{\hat{y}}_{m}; \boldsymbol{\tilde{y}}_{m}\right)
  \leq  \! \frac{T-\tau}{T}\log\!\Bigg[\frac{\lvert\mathbb{E}\{\boldsymbol{\tilde{y}}_{m}\boldsymbol{\tilde{y}}_{m}^{H}\}\!+\!\mathbb{E}\{\boldsymbol{q}_{m}\boldsymbol{q}_{m}^{H}\}\rvert}{\lvert\mathbb{E}\{\boldsymbol{q}_{m}\boldsymbol{q}_{m}^{H}\}\rvert}\Bigg] \!=\! \frac{T-\tau}{T}\log\!\Bigg[\lvert I_{K}\!+\!\boldsymbol{\Psi}_{m}\boldsymbol{Q}_{m}^{-1}\rvert\Bigg]\!\\
  &= \! \frac{T-\tau}{T}\!\sum\limits_{k=1}^{K}\!\!\log\Bigg[1+\frac{\boldsymbol{\Psi}_{m}[k,k]}{{Q}_{mk}}\Bigg],\\
  \end{split}
\end{equation}
where,
\begin{equation}\label{BHECMF3}
 \begin{split}
  \boldsymbol{\Psi}_{m}[k,k] =& \rho_{u}\!\sum\limits_{k^{\prime}=1}^{K}\!\!\eta_{k^{\prime}}\beta_{mk^{\prime}}\gamma_{mk} \!+ \!\! \rho_{u}\frac{1-\xi_{t}}{\tau\xi_{t}}\!\Bigg(\!\sum\limits_{k^{\prime}=1}^{K}\!\!\sqrt{\eta_{k^{\prime}}}\frac{\beta_{mk^{\prime}}}{\beta_{mk}}\gamma_{mk}\!\!\Bigg)^{2} \!\!\!+ \!\! \rho_{u}\frac{\!1\!+\!\xi_{r}\!-\!2\xi_{r}\xi_{t}\!}{\tau\xi_{r}\xi_{t}}\!\sum\limits_{k^{\prime}=1}^{K}\!\!\Big(\!\sqrt{\eta_{k^{\prime}}}\frac{\beta_{mk^{\prime}}}{\beta_{mk}}\gamma_{mk}\Big)^{2}\\
  +& \rho_{u}\eta_{k}\gamma_{mk}^{2} + N\gamma_{mk}.
 \end{split}
\end{equation}
$\boldsymbol{\Psi}_{m}[k,k]$ represents the diagonal elements of $\boldsymbol{\Psi}_{m}=\mathbb{E}\{\boldsymbol{\tilde{y}}_{m}\boldsymbol{\tilde{y}}_{m}^{H}\}$. Off-diagonal elements of $\boldsymbol{\Psi}_{m}$ are zero, since the pilots are orthogonal, and the quantization noises and the channels of UEs are mutually independent. Similar to (\ref{LCEFC}), we propose the following low-complexity scheme to find out the distortion values $Q_{mk}$;
\begin{equation}\label{LCWCF}
  \log\Big[1+\frac{\boldsymbol{\Psi}_{m}[k,k]}{Q_{mk}}\Big] = \frac{\boldsymbol{\Psi}_{m}[k,k]}{\sum_{k=1}^{K}\boldsymbol{\Psi}_{m}[k,k]}\frac{T}{T-\tau}C_{m}.
\end{equation}

\subsubsection{Processing at CU}
After receiving all the quatized vectors transmitted by APs, the CU forms the following vector to recover data of UE\textsubscript{k},
\begin{equation}\label{eqn:WCFatCU}
\boldsymbol{\hat{y}}_{k} = [\tilde{g}_{1k}^{*}y_{1}, \ \tilde{g}_{2k}^{*}y_{2},..., \ \tilde{g}_{mk}^{*}y_{m},..., \ \tilde{g}_{Mk}^{*}y_{M}]^T + [q_{1k},\ q_{2k},...,\ q_{mk},...,\ q_{Mk}]^T.
\end{equation}
By applying UatF one can rewrite (\ref{eqn:WCFatCU}) as follows
\begin{equation}\label{eqn:WCFatCUrewrite}
\boldsymbol{\hat{y}}_{k} = \boldsymbol{b}_{k}s_{k} + \boldsymbol{z}_{k}.
\end{equation}
in which $\boldsymbol{b}_{k}$ and $\boldsymbol{z}_{k}$ are $M\times 1$ vectors where their $m\textsuperscript{th}$-element are defined as $\sqrt{\rho_{u}\eta_{k}\xi_{r}\xi_{t}}\mathbb{E}\{\tilde{g}_{mk}^{*}g_{mk}\}$ and $\hat{y}_{mk}-\sqrt{\rho_{u}\eta_{k}\xi_{r}\xi_{t}}\mathbb{E}\{\tilde{g}_{mk}^{*}g_{mk}\}s_{k}$, respectively, wherein $\hat{y}_{mk}$ is the $m\textsuperscript{th}$-element of $\boldsymbol{\hat{y}}_{k}$.

Afterwards, the CU employs a linear receiver to generate effective received signal of UE\textsubscript{k} as
\begin{equation}\label{eqn:WCFatCUrewrite}
r_k:=\boldsymbol{u}_{k}^{H}\boldsymbol{\hat{y}}_{k} = \boldsymbol{u}_{k}^{H}\boldsymbol{b}_{k}s_{k} + \boldsymbol{u}_{k}^{H}\boldsymbol{z}_{k} = \sum\limits_{m=1}^{M}u_{mk}^{*}(\tilde{g}_{mk}^{*}y_{m} + q_{mk}).
\end{equation}
The linear receiver that maximizes SINR is the MMSE receiver \cite{tse2005fundamentals}. So, $\boldsymbol{u}_{k}^{\text{Opt}} = \mathcal{K}_{z_{k}}^{-1}\boldsymbol{b}_{k}$ is the optimal linear receiver, where $\mathcal{K}_{z_{k}} = \mathbb{E}\{\boldsymbol{z}_{k}\boldsymbol{z}_{k}^{H}\}$.
\begin{theorem}
For the EMCF strategy, the achievable rate of UE\textsubscript{$k$} becomes
\begin{equation}\label{eqn:RateWCF}
  \begin{split}
    R_{k}\!=&\frac{T-\tau}{T}\log_{2}\Big(1+\boldsymbol{b}_{k}^{H}\mathcal{K}_{z_{k}}^{-1}\boldsymbol{b}_{k}\Big),\\
    \mathcal{K}_{z_{k}}[m,n] =& \rho_{u}\xi_{r}\frac{1-\xi_{t}}{\tau\xi_{t}}\sum\limits_{k^{\prime}=1}^{K}\eta_{k^{\prime}}\frac{\beta_{mk^{\prime}}\beta_{nk^{\prime}}}{\beta_{mk}\beta_{nk}}\gamma_{mk}\gamma_{nk}, \ \ \ \{m\neq n\} \in \{1,2,...,M\},\\
    \mathcal{K}_{z_{k}}[m,m] =& \rho_{u}\!\sum\limits_{k^{\prime}=1}^{K}\!\!\eta_{k^{\prime}}\!\Big[\gamma_{mk}\beta_{mk^{\prime}} \!-\!\frac{1}{\tau}\frac{\beta_{mk^{\prime}}^{2}}{\beta_{mk}^{2}}\gamma_{mk}^{2}\!+\!\rho_{p}\beta_{mk^{\prime}}^{2}\lambda_{mk}^{2} \!\Big]\!+\! \rho_{u}\eta_{k}(1-\xi_{r}\xi_{t})\gamma_{mk}^{2}\!+\!N\gamma_{mk}\!+\!Q_{mk},\\
    \boldsymbol{b}_{k} =& \sqrt{\rho_{u}\eta_{k}\xi_{r}\xi_{t}}[\gamma_{1k}, \ \gamma_{2k}, \ ..., \ \gamma_{mk}, \ ..., \ \gamma_{Mk}]^T.
  \end{split}
\end{equation}
where $\mathcal{K}_{z_{k}}[m,n]$ refers to the $m^{\text{th}}$-row and $n^{\text{th}}$-column of the $M\times M$ matrix $\mathcal{K}_{z_{k}}$.
\end{theorem}

\begin{proof}
  The proof is given in Appendix C.
\end{proof}

\section{Optimal Power Allocation for CFE and ECF Strategies}\label{sec4power}
This section focus on maximizing the sum SE (SSE) of the system by controlling the power of each user for data transmission. Basically, this problem improves the performance of the system by reducing the inter-user interference. The optimization problem is as follows,
\begin{equation}\label{power1}
\mathcal{P}_{1}:
  \begin{cases}
    \begin{aligned}
        &\underset{\{\eta_{k}\geq0\}_{k}}{\text{{maximize}}}
        && \frac{T-\tau}{T}\sum\limits_{k=1}^{K}\log(1+\text{SINR}_{k})\\
        &\text{{subject to}} && \eta_{k}\leq 1, \ \forall k
    \end{aligned}
  \end{cases}
\end{equation}
where $\text{SINR}_{k}$ is obtained either by CFE or ECF, which are given in Theorems 1-3. $\mathcal{P}_{1}$ is a NP-hard problem, so we approximate $\mathcal{P}_{1}$ with the following problem $\mathcal{P}_{2}$,
\begin{theorem}
One can obtain the following geometric programming for CFE to efficiently solve the original problem $\mathcal{P}_{1}$,
 \begin{equation}\label{power2}
  \mathcal{P}_{2}:
    \begin{cases}
      \begin{aligned}
        &\underset{\{\eta_{k}\geq0,\ t_{k}\geq0\}_{k}}{\text{{\textit{\emph{maximize}}}}}
                             && \prod\limits_{k=1}^{K}t_{k}\\
        &\text{{\textit{\emph{subject to}}}} && \frac{t_k}{A_{k}\eta_{k}}\sum\limits_{k^{\prime}=1}^{K}\eta_{k^{\prime}}B_{kk^{\prime}} + \frac{L_k}{A_{k}\eta_{k}}\leq 1,\\
        &                    && \eta_{k}\leq 1, \ k = 1,2,...,K
      \end{aligned}
    \end{cases}
 \end{equation}
 where
 \begin{equation}\label{param}
 \begin{aligned}
        A_{k}&=\rho_{k}\xi_{r}\xi_{t}\Gamma_{kk}, \ \ \ \ \mathcal{S}_{kk^{\prime}}=\rho_{u}\sum\limits_{m=1}^{M}\!\left[2^{\frac{T}{T-\tau}C_{d,m}}-1\right]^{-1}\!\!\!\!\gamma_{mk}\beta_{mk^{\prime}}, \\
        B_{kk^{\prime}}&= \rho_{u}\!\!\Bigg[\!\Omega_{kk^{\prime}}\!\!+\!\! \xi_{r}(\!1\!\!-\!\xi_{t}\!)\!\!\left(\!\!\boldsymbol{\varphi}_{k}^{H}\!\boldsymbol{\varphi}_{k^{\prime}}\!\!-\!\frac{1}{\!\tau\xi_{t}\!}\!\right)\!\Gamma_{kk^{\prime}} \!\! + \!\!\rho_{p}(\!1\!\!-\!\xi_{r}\!)\!\Big(\!\tau\xi_{r}\xi_{t}\boldsymbol{\varphi}_{k}^{H}\!\boldsymbol{\varphi}_{k^{\prime}}\!\!-\!(\!1\!+\!\xi_{r}\!\!-\!\xi_{r}\xi_{t}\!)\!\Big)\!\Lambda_{kk^{\prime}}\!\!\Bigg]\! +\! \mathcal{S}_{kk^{\prime}},\\
        L_{k}&=\sum\limits_{m=1}^{M}\!\left(1+\left[2^{\frac{T}{T-\tau}C_{d,m}}-1\right]^{-1}\right)\!N\gamma_{mk}.
 \end{aligned}
 \end{equation}
\end{theorem}
\begin{proof}
  By high SNR approximation, the objective in $\mathcal{P}_{1}$ can be lower-bounded by $\frac{T\!-\!\tau}{T}\!\log\!\!\Big[\!\prod\limits_{k=1}^{K}\!\!\text{SINR}_{k}\!\Big]\!$. Since $\log(x)$ is monotone increasing function of $x$ and $\frac{T-\tau}{T}$ is constant, the objective function can be rewritten as $\prod\limits_{k=1}^{K}\text{SINR}_{k}$. Moreover, by introducing auxiliary variables $t_{k}$, where $\text{SINR}_{k}\!\geq\! t_{k}$, $\mathcal{P}_{1}$ is approximated by $\mathcal{P}_{3}$,
   \begin{equation*}\label{power3}
  \mathcal{P}_{3}:
    \begin{cases}
      \begin{aligned}
        &\underset{\{\eta_{k}\geq0,\ t_{k}\geq0\}_{k}}{\text{{maximize}}}
                             && \prod\limits_{k=1}^{K}t_{k}\\
        &\text{{subject to}} && t_{k}\leq \text{SINR}_{k},\\
        &                    && \eta_{k}\leq 1, \ k = 1,2,...,K
      \end{aligned}
    \end{cases}
 \end{equation*}
 Besides, from Theorem 1, $\text{SINR}_{k}$ for CFE can be rewritten as follows equivalently,
 \begin{equation*}\label{RewrittenCFE}
    \text{SINR}_{k}^{CFE} = \frac{A_{k}\eta_{k}}{\sum\limits_{k^{\prime}=1}^{K}\eta_{k^{\prime}}B_{kk^{\prime}} + L_{k}},
 \end{equation*}
 where, $A_{k}$, $B_{kk^{\prime}}$, and $L_{k}$ are given in (\ref{param}). By further algebraic manipulations $\mathcal{P}_{2}$ could be derived. Therefore, it can be solved efficiently using available solvers such as MOSEK in CVX.
\end{proof}
\setcounter{remark}{4}
\begin{remark}
  By following similar steps as developed for CFE strategy, one can obtain the approximate GP for maximizing the SSE of the ECF given in Theorems 2 and 3.
\end{remark}
\section{Numerical Results}\label{sec5numerical}
Here, numerical results are presented for the three strategies to illustrate the effects of the limited capacity FHLs and hardware non-idealities. To this end, a square area with side $D=1\text{ [km]}$ is considered wherein $K$ UEs and $M$ APs are uniformly and randomly distributed. To avoid boundary effects, the area is  wrapped around.

For the large-scale fading, the following three-slope model is considered similar to \cite{ngo2017cell}
\begin{equation}\label{LSFC}
\begin{split}
  \!\beta_{mk} \ \ \!\! =& PL_{mk} + \sigma_{\text{sh}}z_{mk}\\
 \!P\!L_{mk}\!\!=&
  \begin{cases}
    -L-10\log_{10}\left(d_{mk}^{3.5}\right),& \text{if } d_{mk}\geq d_{1}\\
    -L-10\log_{10}\left(d_{1}^{1.5}d_{mk}^{2}\right),& \text{if } d_{0}< d_{mk} \leq d_{1}\\
    -L-10\log_{10}\left(d_{1}^{1.5}d_{0}^{2}\right),& \text{if } d_{mk} \leq d_{0}
  \end{cases}\\
  \!L\!\! \ \ \ \ \ =& 46.3 + 33.9\log_{10}(f)-13.82\log_{10}(h_{\text{AP}})-(1.1 \log_{10}(f)\!-\!0.7)h_u +
      (1.56 \log_{10}(f)\!-\!0.8),
  \end{split}
\end{equation}
where, $\sigma_{\text{sh}} = 8 \ [\text{dB}]$, $z_{mk}\sim\mathcal{CN}(0,1)$ account for the shadowing, and $P\!L_{mk}$ represents path-loss in [dB]. It is assumed that $h_{\text{AP}} = 15\ [\text{m}], \ h_u = 1.65\ [\text{m}], \ f = 1.9\ [\text{GHz}]$ which are the APs height, user antenna height and carrier frequency, respectively, moreover, $d_{0} = 10\ [\text{m}], \ d_{1} = 50\ [\text{m}]$. We also assume that $T = 200$ samples, $\rho_u = \rho_p = 100\ [\text{mW}]$, and for the noise power, we have
\begin{equation}\label{noise}
  N = B\times k_{B}\times T_{0}\times NF,
\end{equation}

where, $B = 20\ [\text{MHz}]$, $k_{B} = 1.381\times 10^{-23}\ [\text{Joule/Kelvin}]$, $T_{0} = 290\ [\text{Kelvin}]$ and $N\!F = 9\ [\text{dB}]$ are system bandwidth, Boltzmann constant, temperature and noise figure, respectively.
\begin{figure}[!t]
    \centering
        \psfrag{SSE}[][][0.65]{SSE [bits/s/Hz]}
        \psfrag{C}[][][0.65]{$C$ [bits/s/Hz]}
        \psfrag{EMCF}[][][0.6]{ \ \ \ \ \ \ \ \ \ \ \ \ \ \ EMCF: No markers}
        \psfrag{ECFUB}[][][0.6]{ \ \ ECF\textsubscript{UB}: \textcolor[rgb]{0.1992, 0.3984, 0.996}{\scriptsize{$\bigstar$}}}
        \psfrag{CFE}[][][0.6]{ \ \ \!\!\! CFE: \textcolor[rgb]{0.25, 0.25, 0.25}{\scriptsize{\CIRCLE}}}
        \psfrag{EBH}[][][0.6]{ \ \ \ \ \ \ \ \ \ \ \ \ \ \ \ Equal fronthaul: \textcolor[rgb]{0.91,0.59,0.12}{{$\blacktriangle$}}}
        \psfrag{EMCF1}[][][0.6]{\ \ EMCF}
        \psfrag{ECFUB1}[][][0.6]{\ ECF\textsubscript{UB}}
        \psfrag{0.1}[][][0.6]{$0.1$}
        \psfrag{0.2}[][][0.6]{$0.2$}
        \psfrag{0.3}[][][0.6]{$0.3$}
        \psfrag{2}[][][0.6]{$2$}
        \psfrag{6}[][][0.6]{$6$}
        \psfrag{0}[][][0.6]{$0$}
        \psfrag{5}[][][0.6]{$5$}
        \psfrag{10}[][][0.6]{$10$}
        \psfrag{15}[][][0.6]{$15$}
        \psfrag{20}[][][0.6]{$20$}
        \psfrag{25}[][][0.6]{$25$}
        \psfrag{30}[][][0.6]{$30$}
        \psfrag{35}[][][0.6]{$35$}
        \psfrag{-1}[][][0.5]{$-1$}
        \psfrag{1}[][][0.5]{$1$}
        \includegraphics[scale=.4]{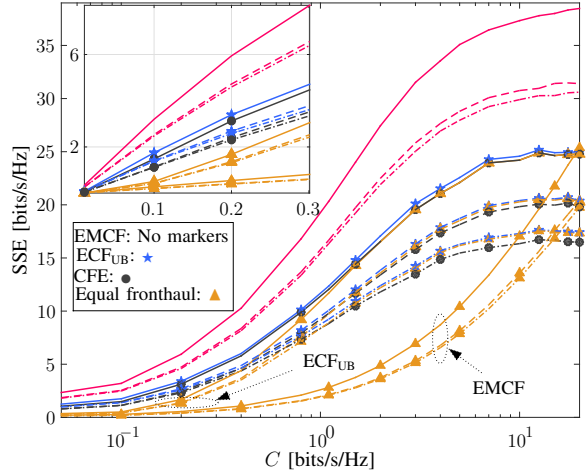}
        \vspace{-0.3cm}
        \caption{{Sum spectral efficiency versus fronthaul capacity for $K=20$ and $M=200$. Solid, dash and dash-dot lines represent perfect hardware, $\{\xi_{r} = 0.8,\ \xi_{t} = 1\}$ and $\{\xi_{r} = 1,\ \xi_{t} = 0.8\}$, respectively.}}\label{fig:SRvsCt}
\vspace{-0.2cm}
\end{figure}%

Fig.~\ref{fig:SRvsCt} studies SSE versus fronthaul capacity, for different  transceiver hardware qualities, and various strategies. With the proposed low-complexity fronthaul allocation, EMCF outperforms both the ECF and CFE strategies. It is worth noting that proposed fronthaul capacity allocation for EMCF substantially improves the SSE, however for ECF the proposed allocation becomes highly advantageous when fronthaul capacity is limited; For instance, it provides 50\% increase in SSE (for all three cases of different transceiver hardware qualities) at $C=0.2[\text{bits/s/Hz}]$ compared to equal fronthaul capacity allocation. Furthermore,  for all strategies, hardware qualities at the UEs are more influential than that of the APs'.
\begin{figure}[!t]
    \centering
        \psfrag{ratio}[][][0.64]{$\frac{\text{SSE\textsubscript{UB}-SSE\textsubscript{LB}}}{\text{SSE\textsubscript{UB}}}$}
        \psfrag{Ztttttttttt}[][][0.64]{$\xi_{t}$}
        \psfrag{Zrrrrrrrrrr}[][][0.64]{$\xi_{r}$}
        \psfrag{0}[][][0.54]{$0$}
        \psfrag{0.9}[][][0.54]{$0.9$}
        \psfrag{0.8}[][][0.54]{$0.8$}
        \psfrag{0.7}[][][0.54]{$0.7$}
        \psfrag{0.6}[][][0.54]{$0.6$}
        \psfrag{0.5}[][][0.54]{$0.5$}
        \psfrag{1}[][][0.54]{$1$}
        \psfrag{0.06}[][][0.54]{$0.06$}
        \psfrag{0.04}[][][0.54]{$0.04$}
        \psfrag{0.02}[][][0.54]{$0.02$}
        \includegraphics[scale=.37]{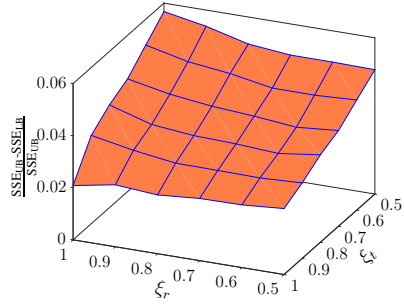}
        \vspace{-0.3cm}
        \caption{{Difference of lower and upper bounds of SSE for ECF with $C=1$ [bits/s/Hz]}}\label{fig:UBLB}
\vspace{-0.2cm}
\end{figure}%

Fig.~\ref{fig:UBLB} represents the difference between upper and lower bounds of SSE with ECF for $C=1$ [bits/s/Hz]. For fairly high quality hardwares, e.g. $0.9\!\leq \!\xi_{t}, \ \xi_{r}\!\leq\! 1$, the maximum SSE difference between the bounds is 3\% which indicates the tightness of the bounds, the difference is even lower for higher fronthaul capacities.

\begin{figure}[t!]
  \centering
        \psfrag{Cp}[][][0.7]{C\textsubscript{$p$} [bits/s/Hz]}
        \psfrag{SSE}[][][0.7]{SSE [bits/s/Hz]}
        \psfrag{WCFxrt1wwwwwwww}[][][0.7]{\!\!\!\!\!\!EMCF, $\xi_{t}=\xi_{r}=1$}
        \psfrag{WCFxirt09wwww}[][][0.7]{\ \ \ \ \ EMCF, $\xi_{t}=\xi_{r}=0.9$}
        \psfrag{ECFUBxirt1wwwwwww}[][][0.7]{\!\!\!\!\!\!\!\!ECF\textsubscript{UB}, $\xi_{t}=\xi_{r}=1$}
        \psfrag{ECFUBxirt09wwwwww}[][][0.7]{ECF\textsubscript{UB}, $\xi_{t}=\xi_{r}=0.9$}
        \psfrag{CFExirt1wwwwwwww}[][][0.7]{\!\!\!\!\!\!\!CFE, $\xi_{t}=\xi_{r}=1$}
        \psfrag{CFExirt09wwwwwwww}[][][0.7]{\!\!\!\!\!\!CFE, $\xi_{t}=\xi_{r}=0.9$}
        \psfrag{C01}[][][0.65]{$\ \ \ \ C=0.1$ [bits/s/Hz]}
        \psfrag{C1}[][][0.65]{$\ \ \ \ \ \ C=1$ [bits/s/Hz]}
        \psfrag{0}[][][0.62]{$0$}
        \psfrag{0.5}[][][0.62]{$0.5$}
        \psfrag{1.5}[][][0.62]{$1.5$}
        \psfrag{2.5}[][][0.62]{$2.5$}
        \psfrag{3.5}[][][0.62]{$3.5$}
        \psfrag{1}[][][0.26]{$1$}
        \psfrag{2}[][][0.62]{$2$}
        \psfrag{3}[][][0.62]{$3$}
        \psfrag{4}[][][0.62]{$4$}
        \psfrag{5}[][][0.62]{$5$}
        \psfrag{10}[][][0.62]{$10$}
        \psfrag{15}[][][0.62]{$15$}
        \psfrag{20}[][][0.62]{$20$}
        \psfrag{-1}[][][0.53]{$\ -1$}
        \psfrag{-2}[][][0.53]{$\ -2$}
        \psfrag{-3}[][][0.53]{$\ -3$}
        \psfrag{-4}[][][0.53]{$\ -4$}
        \psfrag{-5}[][][0.53]{$\ -5$}
  \includegraphics[scale=.52]{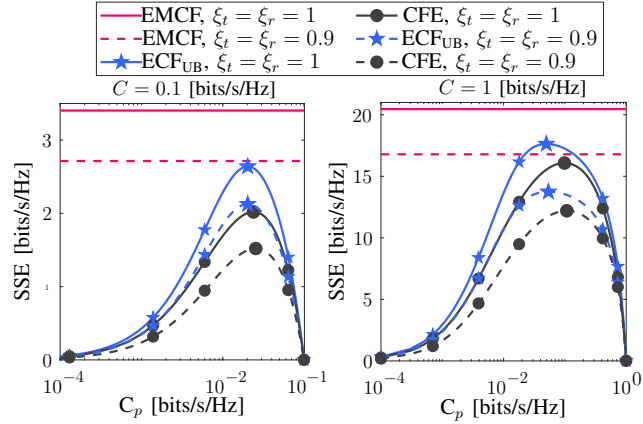}
    \vspace{-0.3cm}
  \caption{Sum spectral efficiency versus allocated fronthaul capacity for pilot transmission.}\label{fig:SSEvsCp2}
    \vspace{-0.2cm}
\end{figure}

\begin{figure}[t!]
  \centering
        \psfrag{SNR}[][][0.72]{SNR [dB]}
        \psfrag{AAAAAAAAAAAP1}[][][0.72]{\!\!\!AP,\ $\xi_r=\xi_t=1$}
        \psfrag{AAAAAAAAAAAP3}[][][0.72]{\!\!AP,\ $\xi_r=\xi_t=0.3$}
        \psfrag{AAAAAAAAAAAP7}[][][0.72]{AP,\ $\xi_r=\xi_t=0.7$}
        \psfrag{CCCCCCCCCCCU1}[][][0.72]{CU,\ $\xi_r=\xi_t=1$}
        \psfrag{CCCCCCCCCCCU3}[][][0.72]{CU,\ $\xi_r=\xi_t=0.3$}
        \psfrag{CCCCCCCCCCCU7}[][][0.72]{CU,\ $\xi_r=\xi_t=0.7$}
        \psfrag{HighSNR}[][][0.65]{High SNR}
        \psfrag{Approximation}[][][0.65]{approximation}
        \psfrag{CCCCCCCCCCCU7}[][][0.72]{CU,\ $\xi_r=\xi_t=0.7$}
        \psfrag{C01}[][][0.62]{$\ \ \ \ \ \ \ \ \ \ \ \ \text{C}=0.1$ [bits/s/Hz]}
        \psfrag{C1}[][][0.62]{$\ \ \ \ \text{C}=1$ [bits/s/Hz]}
        \psfrag{RelativeMSE}[][][0.7]{Relative MSE}
        \psfrag{0}[][][0.62]{$0$}
        \psfrag{1}[][][0.62]{$1$}
        \psfrag{20}[][][0.62]{$20$}
        \psfrag{40}[][][0.62]{$40$}
        \psfrag{-20}[][][0.62]{$-20$}
        \psfrag{10}[][][0.62]{$10$}
        \psfrag{2}[][][0.62]{$2$}
        \psfrag{3}[][][0.62]{$3$}
        \psfrag{4}[][][0.62]{$4$}
        \psfrag{5}[][][0.62]{$5$}
        \psfrag{-3}[][][0.53]{$-3$}
        \psfrag{-4}[][][0.53]{$-4$}
  \includegraphics[scale=.52]{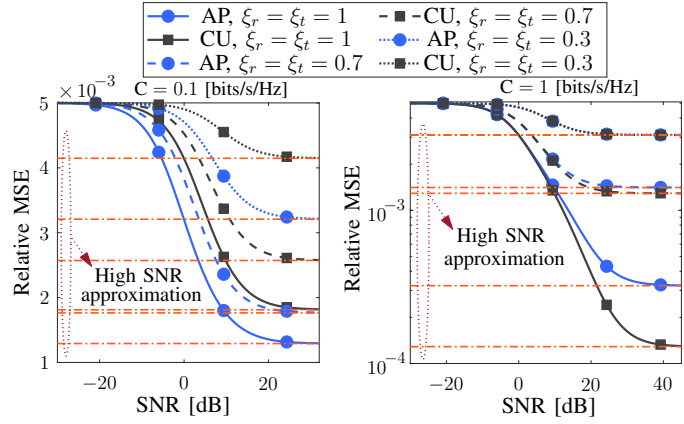}
    \vspace{-0.3cm}
  \caption{Ultimate channel estimation error at CU versus SNR for two schemes of the estimation at AP and CU with  different HIs.}\label{fig:MSEvsSNR}
    \vspace{-0.2cm}
\end{figure}

Fig.~\ref{fig:SSEvsCp2} shows the SSE versus the allocated fronthaul capacity for pilot transmission, i.e. $C_p$, for $C_{m}=0.1$ and $C_{m}=1$. This figure pinpoints the fact that although allocating large portion of the fronthaul capacity for CSI transmission improves the CSI estimation quality, it causes more distortion on the quantized data signals and results in achievable rate degradation. While for small values of $C_p$, inferior channel estimation quality is the main reason for performance degradation.

Fig.~\ref{fig:MSEvsSNR} compares the channel estimation error when the channel estimation is carried out at the CU or AP for $C_{m} = 0.1$ and $C_{m}=1$. We define relative MSE per AP as $\frac{\beta - \mathbb{E}\{|\hat{g}|^2\}}{{\text{mean}(\beta)}}$ and $\text{SNR} = 10\log\left(\frac{\rho}{\sigma^2}\text{mean}(\beta)\right)$. As the figure shows, for lower fronthaul capacities, performing the channel estimation at the AP could lead to smaller estimation errors while for higher fronthaul capacities it is favorable to conduct channel estimation at the CU. Also, hardware impairments and fronthaul capacities are the limiting factor at high SNR regime, such that the estimation error at high SNR regime becomes non-zero for finite fronthaul capacities and imperfect hardwares.
\begin{figure}[t!]
  \centering
        \psfrag{Cd}[][][0.75]{$C_{d}^{*}$ [bits/s/Hz]}
        \psfrag{HI}[][][0.8]{$\xi_r = \xi_t$}
        \psfrag{ECFUBB}[][][0.6]{ECF\textsubscript{UB}}
        \psfrag{ECFLBB}[][][0.6]{ECF\textsubscript{LB}}
        \psfrag{CFE}[][][0.6]{\ \ \ CFE}
        \psfrag{Ct1}[][][0.6]{C\textsubscript{$m$} $\!=\! 1\!$ [bits/s/Hz]}
        \psfrag{0.2}[][][0.7]{$0.2$}
        \psfrag{0.4}[][][0.7]{$0.4$}
        \psfrag{0.6}[][][0.7]{$0.6$}
        \psfrag{0}[][][0.7]{$0$}
        \psfrag{0.8}[][][0.7]{$0.8$}
        \psfrag{0.7}[][][0.7]{$0.7$}
        \psfrag{0.75}[][][0.7]{$0.75$}
        \psfrag{0.85}[][][0.7]{$0.85$}
        \psfrag{1}[][][0.7]{$1$}
        \psfrag{0.9}[][][0.7]{$0.9$}
        \psfrag{0.95}[][][0.7]{$0.95$}
  \includegraphics[scale=.46]{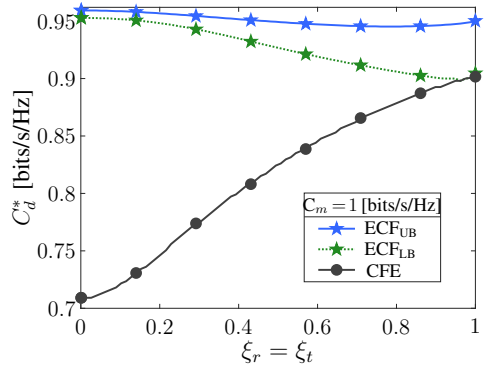}
    \vspace{-0.3cm}
  \caption{Optimum allocated fronthaul capacity for data transmission versus HIs.}\label{fig:CdvsHI}
    \vspace{-0.2cm}
\end{figure}

In Fig.~\ref{fig:CdvsHI} the optimal capacity for data transmission, $C_{d}^{*}$, versus transceiver HIs for $C_{m}=1$ [bits/s/Hz] is depicted. For CFE, by improving the hardware qualities more fronthaul capacity is dedicated for data transmission than CSI transmission, while for the ECF, $C_{d}^{*}$ does not change substantially by varying hardware qualities.

\begin{figure}[t!]
  \centering
        \psfrag{CDF}[][][0.77]{CDF}
        \psfrag{PUSE}[][][0.7]{Per User SE}
        \psfrag{SSE}[][][0.7]{SSE}
        \psfrag{SC1AAAAAAAAAAAAAAAA}[][][0.65]{\!\!\!\!\!\!\!\!\!\!\!\!\!\!\!\!\!\!EMCF,\ $\xi_r=\xi_t=1$}
        \psfrag{ECFUBPC1AAAAAAAAAAAA}[][][0.65]{\!\!\!\!\!\!\!\!\!\!\!\!ECF\textsubscript{UB,Opt},\ $\xi_r=\xi_t=1$}
        \psfrag{CFEPC1AAAAAAAAAAAAAAA}[][][0.65]{\!\!\!\!\!\!\!\!\!\!\!\!\!\!\!\!\!\!\!\!\!CFE\textsubscript{Opt},\ $\xi_r=\xi_t=1$}
        \psfrag{SC09AAAAAAAAAAAAAAAAA}[][][0.65]{\!\!\!\!\!\!\!\!\!\!\!\!\!\!\!\!\!\!EMCF,\ $\xi_r=\xi_t=0.9$}
        \psfrag{ECFUBPC09AAAAAAAAAAAAA}[][][0.65]{\!\!\!\!\!\!\!\!\!\!\!ECF\textsubscript{UB,Opt},\ $\xi_r=\xi_t=0.9$}
        \psfrag{ECFUB1AAAAAAAAAAAA}[][][0.65]{\!\!\!\!\!\!\!\!\!ECF\textsubscript{UB},\ $\xi_r=\xi_t=1$}
        \psfrag{CFE1AAAAAAAAAAAAAA}[][][0.65]{\!\!\!\!\!\!\!\!\!\!\!\!\!\!\!CFE,\ $\xi_r=\xi_t=1$}
        \psfrag{CFEUB09AAAAAAAAAAAA}[][][0.65]{\!\!\!\!\!\!\!ECF\textsubscript{UB},\ $\xi_r=\xi_t=0.9$}
        \psfrag{CFEPC09AAAAAAAAAAAA}[][][0.65]{\!\!\!\!\!\!\!\!\!\! CFE\textsubscript{Opt},\ $\xi_r\!=\!\xi_t\!=\!0.9$}
        \psfrag{CFE09AAAAAAAAAAAAAA}[][][0.65]{\!\!\!\!\!\!\!\!\!\!\!\!CFE,\ $\xi_r=\xi_t=0.9$}
        \psfrag{0}[][][0.62]{$0$}
        \psfrag{5}[][][0.62]{$5$}
        \psfrag{10}[][][0.62]{$10$}
        \psfrag{15}[][][0.62]{$15$}
        \psfrag{20}[][][0.62]{$20$}
        \psfrag{0.5}[][][0.62]{$0.5$}
        \psfrag{1}[][][0.62]{$1$}
        \psfrag{1.5}[][][0.62]{$1.5$}
        \psfrag{2}[][][0.62]{$2$}
        \psfrag{0.2}[][][0.62]{$0.2$}
        \psfrag{0.4}[][][0.62]{$0.4$}
        \psfrag{0.6}[][][0.62]{$0.6$}
        \psfrag{0.8}[][][0.62]{$0.8$}
  \includegraphics[scale=.5]{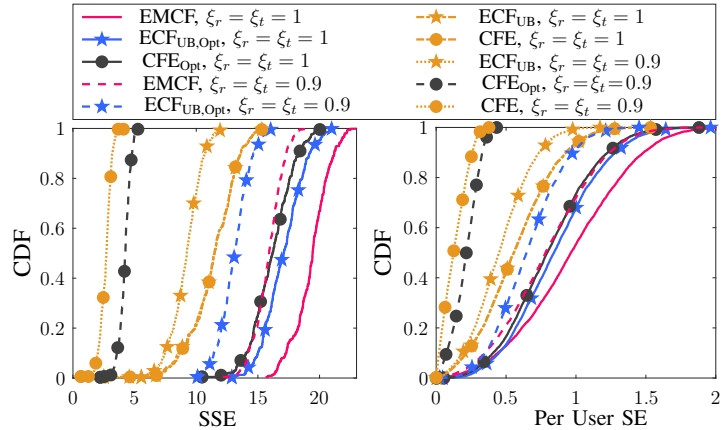}
    \vspace{-0.3cm}
  \caption{CDF of SSE and per UE's SE for $M=200$, $K=20$, and $C_{m}=1$ [bits/s/Hz].}\label{fig:puSESSE}
    \vspace{-0.2cm}
\end{figure}

Fig.~\ref{fig:puSESSE} addresses the cumulative distribution function (CDF) for per user SE and SSE. This figure highlights the performance improvements by optimizing the power and fronhaul allocation compared to full power and equal fronthaul allocation for CFE and ECF. As it is illustrated, optimizing $C_{p}$ and $C_{d}$ along with power allocation improves the SSE of the system between $60\%$ to $90\%$ for the $5\%$-outage sum rate, in particular. It also improves the performance of each user significantly.

Fig.~\ref{fig:EEvsSE} depicts EE as a function of SSE for $C_m = 1$ [bits/s/Hz], $K=20$, $M=100$, and two set of HIs; ${\xi_{r}=\xi_{t}=1}$ and ${\xi_{r}=\xi_{t}=0.9}$. The EE of the system is given by
\begin{equation}\label{EE}
  EE = \frac{B\sum_{k=1}^{K}R_{k}}{P_{t}},
\end{equation}
where $P_{t}$ indicates the total power consumption, such that,
\begin{equation}\label{PC}
  P_{t} = \sum\limits_{k=1}^{K}P_{k} + \sum\limits_{m=1}^{M}P_{m} + B\sum\limits_{m=1}^{M}C_{m}P_{bh,m},
\end{equation}
where $P_{m}=0.2$ [W] and $P_{bh,m}=0.25$ [Watt/Gbit/s] denote power consumption by the fronthaul link and AP\textsubscript{$m$}, respectively, and $P_{k}$  represents the power consumption by UE\textsubscript{$k$}. In either cases, EMCF in comparison with other two strategies can improve the EE of the system substantially.
\begin{figure}[!t]
\vspace{-0.4cm}
\centering
    \begin{minipage}[b]{.45\textwidth}
    \centering
        \psfrag{EE}[][][0.65]{EE [Mbits/J]}
        \psfrag{SSE}[][][0.65]{SSE [bits/s/Hz]}
        \psfrag{zrt1}[][][0.65]{$\xi_{t}=\xi_{r}=1$}
        \psfrag{zrt09}[][][0.65]{$\xi_{t}=\xi_{r}=0.9$}
        \psfrag{EMCF}[][][0.53]{ \ \ \ \ \ \ \ \ \ \ \ \ \ \ \ \ EMCF: No markers}
        \psfrag{ECFUB}[][][0.53]{ \ \ \ ECF\textsubscript{UB}: \textcolor[rgb]{0.1992, 0.3984, 0.996}{\scriptsize{$\bigstar$}}}
        \psfrag{CFE}[][][0.53]{ \ \ \! CFE: \textcolor[rgb]{0.25, 0.25, 0.25}{\scriptsize{\CIRCLE}}}
        \psfrag{5}[][][0.6]{$5$}
        \psfrag{10}[][][0.6]{$10$}
        \psfrag{15}[][][0.6]{$15$}
        \psfrag{1}[][][0.56]{$1$}
        \psfrag{0}[][][0.56]{$0$}
        \psfrag{-1}[][][0.56]{$-1$}
  \includegraphics[scale=.32]{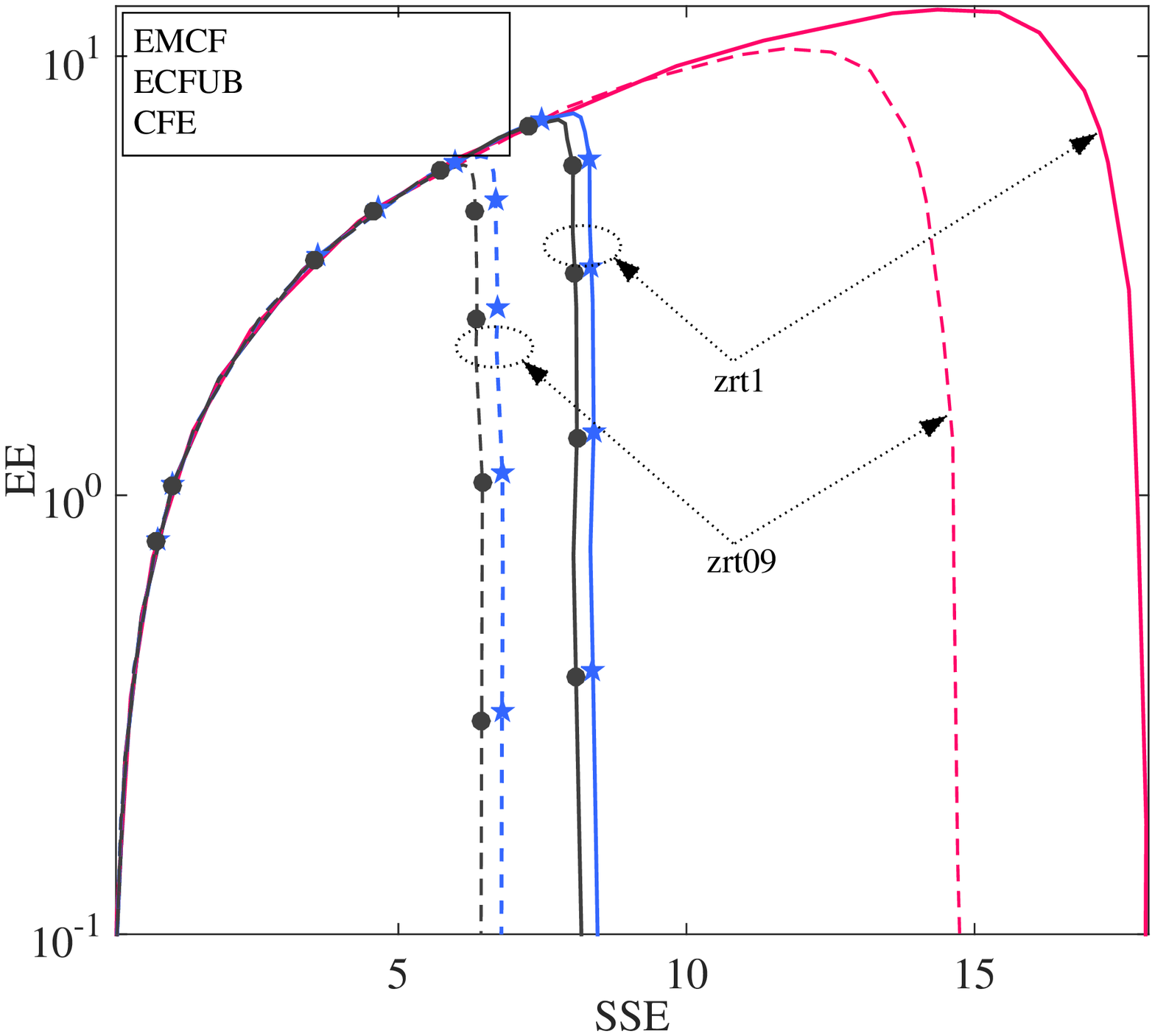}
    \vspace{-0.3cm}
  \caption{Energy efficiency versus SSE.}\label{fig:EEvsSE}
  \vspace{-0.2cm}
\end{minipage}%
\hfill
\begin{minipage}[b]{.45\textwidth}
\vspace{-2cm}
  \centering
        \psfrag{EE}[][][0.65]{EE [Mbits/J]}
        \psfrag{C}[][][0.65]{$C$ [bits/s/Hz]}
        \psfrag{zrt1}[][][0.65]{$\xi_{t}=\xi_{r}=1$}
        \psfrag{zrt09}[][][0.65]{$\xi_{t}=\xi_{r}=0.9$}
        \psfrag{EMCF}[][][0.53]{ \ \ \ \ \ \ \ \ \ \ \ \ \ \ \ \ EMCF: No markers}
        \psfrag{ECFUB}[][][0.53]{ \ \ \ ECF\textsubscript{UB}: \textcolor[rgb]{0.1992, 0.3984, 0.996}{\scriptsize{$\bigstar$}}}
        \psfrag{CFE}[][][0.53]{ \ \ \! CFE: \textcolor[rgb]{0.25, 0.25, 0.25}{\scriptsize{\CIRCLE}}}
        \psfrag{2}[][][0.58]{$2$}
        \psfrag{4}[][][0.6]{$4$}
        \psfrag{6}[][][0.6]{$6$}
        \psfrag{8}[][][0.6]{$8$}
        \psfrag{10}[][][0.6]{$10$}
        \psfrag{12}[][][0.6]{$12$}
        \psfrag{14}[][][0.6]{$14$}
        \psfrag{16}[][][0.6]{$16$}
        \psfrag{-1}[][][0.56]{$-1$}
        \psfrag{-2}[][][0.56]{$-2$}
        \psfrag{0}[][][0.57]{$0$}
        \psfrag{1}[][][0.57]{$1$}
  \includegraphics[scale=.32]{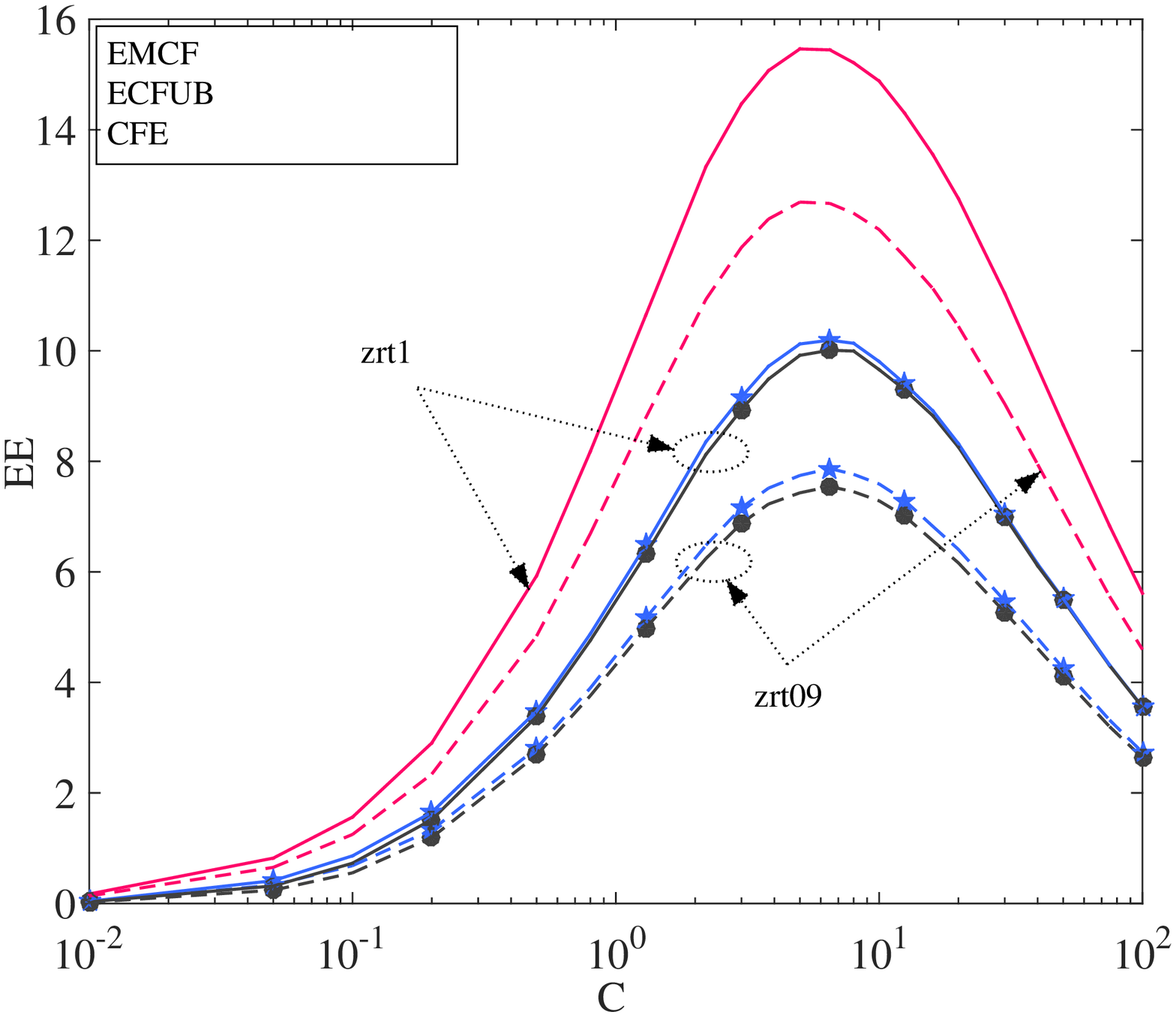}
    \vspace{-0.3cm}
  \caption{Energy efficiency versus fronthaul capacity.}\label{fig:EEvsC}
    \vspace{-0.2cm}
     \end{minipage}
    \vspace{-0.7cm}
\end{figure}



Fig.~\ref{fig:EEvsC} illustrates the EE as a function of fronthaul capacity for $K=20$ and $M=200$. For the strictly limited capacity fronthaul links, the lower data rate due to the limited fronthaul capacity is the major reason for low EE. On the other hand, as the fronthaul capacity increases the consumed power also increases which reduces the EE.

\begin{figure}[!t]
\centering
        \psfrag{SSE}[][][0.7]{SSE [bits/s/Hz]}
        \psfrag{C}[][][0.7]{$C$ [bits/s/Hz]}
        \psfrag{NoAP}[][][0.7]{Number of APs}
        \psfrag{EMCFxirt1}[][][0.61]{\ \ \ \ \ \ \ \ EMCF $\!\xi_{r}\!=\!\xi_{t}\!=\!1\!$}
        \psfrag{ECFUBxirt1}[][][0.61]{\!\!\! ECF\textsubscript{UB} $\!\xi_{r}\!=\!\xi_{t}\!=\!1\!$}
        \psfrag{ECFxirt1}[][][0.61]{\ \ \ \ \ \ \ \ \ CFE $\!\xi_{r}\!=\!\xi_{t}\!=\!1\!$}
        \psfrag{EMCFxirt09}[][][0.61]{EMCF $\!\xi_{r}\!=\!\xi_{t}\!=\!0.9\!$}
        \psfrag{ECFUBxirt09}[][][0.61]{\ \ \ ECF\textsubscript{UB} $\!\xi_{r}\!=\!\xi_{t}\!=\!0.9\!$}
        \psfrag{CFExirt09}[][][0.61]{\ CFE $\!\xi_{r}\!=\!\xi_{t}\!=\!0.9\!$}
        \psfrag{0}[][][0.65]{$0$}
        \psfrag{20}[][][0.65]{$20$}
        \psfrag{40}[][][0.65]{$40$}
        \psfrag{60}[][][0.65]{$60$}
        \psfrag{80}[][][0.65]{$80$}
        \psfrag{25}[][][0.65]{$25$}
        \psfrag{5}[][][0.65]{$5$}
        \psfrag{10}[][][0.65]{$10$}
        \psfrag{15}[][][0.65]{$15$}
        \subfloat[]{\includegraphics[scale=0.34]{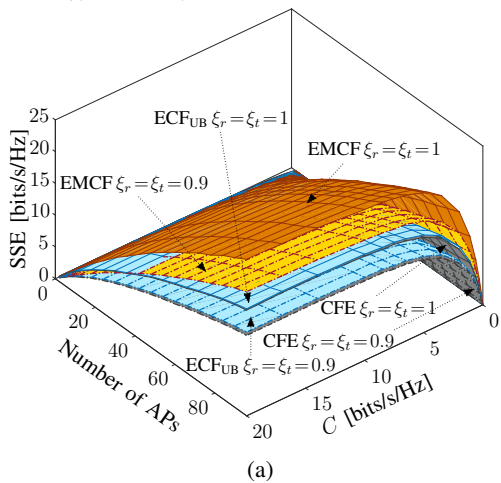}
\label{fig:SEvsMCt}}
\hfil
\psfrag{EE}[][][0.7]{EE [Mbits/J]}
        \psfrag{NoAP}[][][0.7]{Number of APs}
        \psfrag{C}[][][0.7]{$C$ [bits/s/Hz]}
        \psfrag{EMCFxirt1}[][][0.61]{\ \ \ \ \ EMCF $\!\xi_{r}\!=\!\xi_{t}\!=\!1\!$}
        \psfrag{ECFUBxirt1}[][][0.61]{\ \ \ \ \ \ ECF\textsubscript{UB} $\!\xi_{r}\!=\!\xi_{t}\!=\!1\!$}
        \psfrag{CFExirt1}[][][0.61]{\ \ \ \ \ \ CFE $\!\xi_{r}\!=\!\xi_{t}\!=\!1\!$}
        \psfrag{EMCFxirt09}[][][0.61]{\ \ EMCF $\!\xi_{r}\!=\!\xi_{t}\!=\!0.9\!$}
        \psfrag{ECFUBxirt09}[][][0.61]{\ \ \ \ \ \ \ ECF\textsubscript{UB} $\!\xi_{r}\!=\!\xi_{t}\!=\!0.9\!$}
        \psfrag{CFExirt09}[][][0.61]{CFE $\!\xi_{r}\!=\!\xi_{t}\!=\!0.9\!$}
        \psfrag{0}[][][0.65]{$0$}
        \psfrag{20}[][][0.65]{$20$}
        \psfrag{50}[][][0.65]{$50$}
        \psfrag{100}[][][0.65]{$100$}
        \psfrag{25}[][][0.65]{$25$}
        \psfrag{5}[][][0.65]{$5$}
        \psfrag{10}[][][0.65]{$10$}
        \psfrag{15}[][][0.65]{$15$}
\subfloat[]{\includegraphics[scale=0.34]{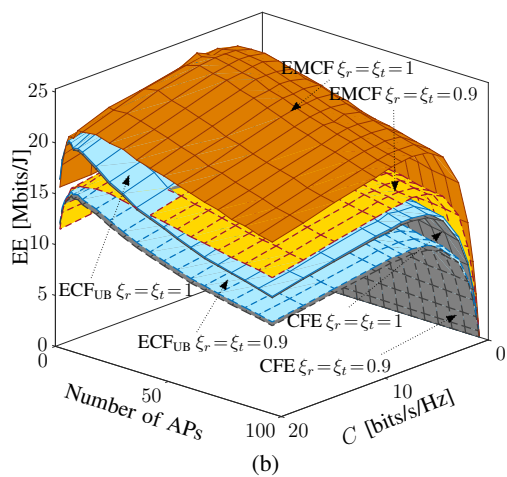}
\label{fig:EEvsMCt}}
\hfil
\caption{Impact of the number of APs and fronthaul capacity on: (a) SSE, and (b) EE.}
\label{fig:SEEEvsMCt}
\end{figure}

\begin{figure}[!t]
\centering
        \psfrag{SE}[][][0.7]{SSE [bits/s/Hz]}
        \psfrag{NoUE}[][][0.7]{Number of UEs}
        \psfrag{NoAP}[][][0.7]{Number of APs}
        \psfrag{EMCFxirt1}[][][0.61]{\ \ \ \ \ \ \ \ EMCF $\!\xi_{r}\!=\!\xi_{t}\!=\!1\!$}
        \psfrag{ECFUBxirt1}[][][0.61]{\!\!\! ECF\textsubscript{UB} $\!\xi_{r}\!=\!\xi_{t}\!=\!1\!$}
        \psfrag{ECFxirt1}[][][0.61]{\ \ \ \ \ \ \ \ \ CFE $\!\xi_{r}\!=\!\xi_{t}\!=\!1\!$}
        \psfrag{EMCFxirt09}[][][0.61]{\ \ \ EMCF $\!\xi_{r}\!=\!\xi_{t}\!=\!0.9\!$}
        \psfrag{ECFUBxirt09}[][][0.61]{\ \ \ ECF\textsubscript{UB} $\!\xi_{r}\!=\!\xi_{t}\!=\!0.9\!$}
        \psfrag{CFExirt09}[][][0.61]{\ \ \ \ \ \ \ \ \ CFE $\!\xi_{r}\!=\!\xi_{t}\!=\!0.9\!$}
        \psfrag{CFExirt1}[][][0.61]{\ CFE $\!\xi_{r}\!=\!\xi_{t}\!=\!1\!$}
        \psfrag{0}[][][0.65]{$0$}
        \psfrag{20}[][][0.65]{$20$}
        \psfrag{40}[][][0.65]{$40$}
        \psfrag{60}[][][0.65]{$60$}
        \psfrag{80}[][][0.65]{$80$}
        \psfrag{100}[][][0.65]{$100$}
        \psfrag{50}[][][0.65]{$50$}
        \psfrag{5}[][][0.65]{$5$}
        \psfrag{10}[][][0.65]{$10$}
        \psfrag{15}[][][0.65]{$15$}
        \subfloat[]{\includegraphics[scale=0.34]{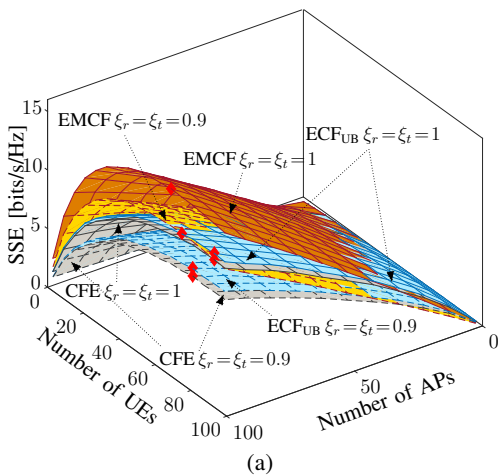}
\label{fig:SEvsUeAp}}
\hfil
        \psfrag{EE}[][][0.7]{EE [Mbits/J]}
        \psfrag{NoAP}[][][0.7]{Number of APs}
        \psfrag{NoUE}[][][0.7]{Number of UEs}
        \psfrag{EMCFxirt1}[][][0.61]{\ \ \ \ \ \ EMCF $\!\xi_{r}\!=\!\xi_{t}\!=\!1\!$}
        \psfrag{ECFUBxirt1}[][][0.61]{\ \ \ \ \ \ ECF\textsubscript{UB} $\!\xi_{r}\!=\!\xi_{t}\!=\!1\!$}
        \psfrag{CFExirt1}[][][0.61]{\ \ \ \ \ \ CFE $\!\xi_{r}\!=\!\xi_{t}\!=\!1\!$}
        \psfrag{EMCFxirt09}[][][0.61]{\ \ EMCF $\!\xi_{r}\!=\!\xi_{t}\!=\!0.9\!$}
        \psfrag{ECFUBxirt09}[][][0.61]{\ \ \ \ \ \ \ ECF\textsubscript{UB} $\!\xi_{r}\!=\!\xi_{t}\!=\!0.9\!$}
        \psfrag{CFExirt09}[][][0.61]{CFE $\!\xi_{r}\!=\!\xi_{t}\!=\!0.9\!$}
        \psfrag{0.1}[][][0.65]{$2$}
        \psfrag{0.2}[][][0.65]{$4$}
        \psfrag{0.3}[][][0.65]{$6$}
        \psfrag{0.4}[][][0.65]{$8$}
        \psfrag{0.5}[][][0.65]{$10$}
        \psfrag{0.6}[][][0.65]{$12$}
        \psfrag{0.7}[][][0.65]{$14$}
        \psfrag{0}[][][0.65]{$0$}
        \psfrag{50}[][][0.65]{$50$}
        \psfrag{100}[][][0.65]{$100$}
\subfloat[]{\includegraphics[scale=0.34]{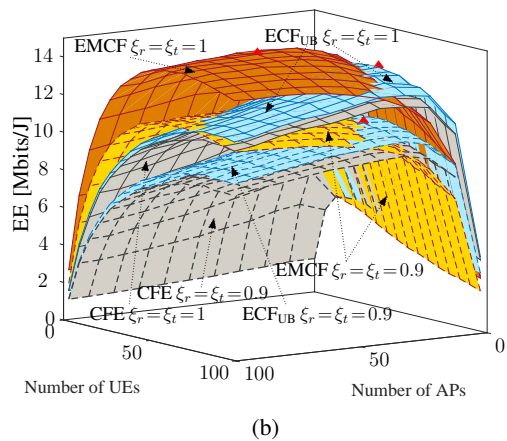}
\label{fig:EEvsUeAp}}
\hfil
\caption{Impact of the number of APs and UEs on (a) SSE, and (b) EE.}
\label{fig:SEEEvsUeAp}
\end{figure}

Fig.~12(a) and Fig.~12(b) investigate the joint effect of the fronthaul capacity and number of the APs on the SSE and EE of the system for $K=20$. Increasing number of the APs or fronthaul capacity always improves the SSE however from the EE viewpoint increasing number of the APs or capacity of the fronthaul links does not necessarily enhance the EE due to the increase in system power consumption.

Fig.~13(a) and Fig.~13(b) shows the SSE and EE as a function of the number of UEs and APs for $C=1$. As depicted increasing the number of antennas improve the SE constantly however, in terms of the EE this is not the case and EE reaches its maximal point for a moderate number of the APs. Also, there is optimal number of UEs to serve by the CF-mMIMO system.

\section{Conclusions}\label{sec6conclude}

We considered the uplink scenario of a fronthaul-constrained CF-mMIMO system in presence of non-negligible residual hardware impairments at the UEs and APs. To manage the limitations due to the finite capacity fronthaul link, three strategies were employed. Closed-form expressions for achievable data rates with EMCF and CFE were derived, and upper and lower bounds for ECF were proposed which were tight enough especially for perfect hardware and FHLs. Besides, by use of the proposed low-complexity fronthaul capacity allocation, EMCF outperformed the other two strategies. Moreover, it was proven that at hight SNR regime and for high fronthaul capacity, estimating channels at CU could result in lower estimation errors than that of APs. For a single user CF-mMIMO, the result indicated that at high SNR, ECF strategy outperformed CFE one. Also GP power control was developed to improve the SSE of the system for CFE and ECF strategies. Finally, the sum SE and EE of the system were studied through the numerical results to highlight the performance characteristics of the system.


\section*{Appendix A}\label{app:CFE}
\vspace{-.41cm}
After sequences of mathematical manipulations it can be shown that
\begin{equation*}\label{eqn:DSkCFE}
\begin{split}
  {\!\lvert\text{DS}_{k}\rvert}^2 = \rho_{u}\eta_{k}\xi_{r}\xi_{t}\left(\!\sum\limits_{m=1}^{M}\!\!\gamma_{mk}\!\!\right)^2\!, \ \ \ \ \ \mathbb{E}\!\left\lbrace{\!\lvert\text{RN}_{k}\rvert}^2\right\rbrace\!\! = \!\!N\sum\limits_{m=1}^{M}\!\!\gamma_{mk}, \ \ \ \ \ \mathbb{E}\!\left\lbrace{\!\lvert\text{QN}_{k}\rvert}^2\right\rbrace\!\! =\!\! \sum\limits_{m=1}^{M}\!\!Q_{d,m}\gamma_{mk},
  \end{split}
\end{equation*}
\vspace{-0.6cm}
\begin{equation*}
\begin{split}
  \mathbb{E}\left\lbrace{\!\lvert\text{BU}_{k}\rvert}^2\right\rbrace = \rho_{u}\eta_{k}\xi_{r}\xi_{t}\Bigg[\sum\limits_{m=1}^{M}\gamma_{mk}\beta_{mk} + \frac{1-\xi_{t}}{\tau\xi_{t}}\left(\sum\limits_{m=1}^{M}\gamma_{mk}\right)^2 + \rho_{p}(1-\xi_{r})\sum\limits_{m=1}^{M}\lambda_{mk}^{2}\beta_{mk}^{2}\Bigg],
  \end{split}
\end{equation*}
\vspace{-0.5cm}
\begin{equation*}
\begin{split}
  \mathbb{E}\left\lbrace{\!\lvert\text{IUI}_{kk^{\prime}}\rvert}^2\right\rbrace\! = \! \rho_{u}\eta_{k^{\prime}}\xi_{r}\xi_{t}\!\Bigg[\!\sum\limits_{m=1}^{M}\gamma_{mk}\beta_{mk^{\prime}}\! +\! \frac{1-\xi_{t}}{\tau\xi_{t}}\!\left(\sum\limits_{m=1}^{M}\gamma_{mk}\frac{\beta_{mk^{\prime}}}{\beta_{mk}}\!\right)^2\!\!\!\!+\! \rho_{p}(1-\xi_{r})\sum\limits_{m=1}^{M}\lambda_{mk}^{2}\beta_{mk^{\prime}}^{2}\!\Bigg],
  \end{split}
\end{equation*}
\vspace{-0.5cm}
\begin{equation*}
  \begin{split}
\hspace{-1cm}     \mathbb{E}\left\lbrace{\!\lvert\text{THI}_{kk^{\prime}}\rvert}^2\right\rbrace =& \rho_{u}\eta_{k^{\prime}}\xi_{r}(1-\xi_{t})\Bigg[\sum\limits_{m=1}^{M}\gamma_{mk}\beta_{mk^{\prime}} + \left(\boldsymbol{\varphi}_{k}^{H}\boldsymbol{\varphi}_{k^{\prime}} + \frac{1-\xi_{t}}{\tau\xi_{t}}\right)\!\!\left(\sum\limits_{m=1}^{M}\gamma_{mk}\frac{\beta_{mk^{\prime}}}{\beta_{mk}}\right)^2 \\
       & + \rho_{p}(1-\xi_{r})\sum\limits_{m=1}^{M}\lambda_{mk}^{2}\beta_{mk^{\prime}}^{2}\Bigg],
  \end{split}
\end{equation*}
\vspace{-0.2cm}
\begin{equation*}
  \begin{split}
\hspace{-0cm}     \mathbb{E}\left\lbrace{\!\lvert\text{RHI}_{k}\rvert}^2\right\rbrace =& \sum\limits_{k^{\prime} = 1}^{K}\rho_{u}\eta_{k^{\prime}}(1-\xi_{r})\Bigg[\sum\limits_{m=1}^{M}\gamma_{mk}\beta_{mk^{\prime}} + \rho_{p}\xi_{r}\Big(\tau\xi_{t}\lvert\boldsymbol{\varphi}_{k}^{H}\boldsymbol{\varphi}_{k^{\prime}}\rvert + (1-\xi_{t})\Big)\!\!\sum\limits_{m=1}^{M}\lambda_{mk}^{2}\beta_{mk^{\prime}}^{2} \\
       & + \rho_{p}(1-\xi_{r})\sum\limits_{m=1}^{M}\lambda_{mk}^{2}\beta_{mk^{\prime}}^{2}\Bigg].
  \end{split}
\end{equation*}
\section*{Appendix B}\label{app:ECF_UpperBound}
\vspace{-.41cm}
By plugging in $\hat{g}_{mk}^{*}$ instead of $\tilde{g}_{mk}^{*}$ in (\ref{eqn:MRC_CFE}), and after sequences of mathematical manipulations it can be shown that
\begin{equation*}
\begin{split}
  {\!\lvert\text{DS}_{k}\rvert}^2 = \rho_{u}\eta_{k}\xi_{r}\xi_{t}\left(\!\sum\limits_{m=1}^{M}\!\!\gamma_{mk}^{\prime
  }\!\!\right)^2\!, \ \ \ \ \ \mathbb{E}\!\left\lbrace{\!\lvert\text{RN}_{k}\rvert}^2\right\rbrace\!\! = \!\!N\sum\limits_{m=1}^{M}\!\!\gamma_{mk}^{\prime
  }, \ \ \ \ \ \mathbb{E}\!\left\lbrace{\!\lvert\text{QN}_{k}\rvert}^2\right\rbrace\!\! =\!\! \sum\limits_{m=1}^{M}\!\!Q_{d,m}\gamma_{mk}^{\prime},
  \end{split}
\end{equation*}

Above obtained relations have exact values. For other parts we obtain upper bounds. Hence, computed rate will be a lower-bound on the achievable rate.
\begin{equation}\label{BUup}
 \mathbb{E}\!\left\lbrace{\!\lvert\text{BU}_{k}\rvert}^2\right\rbrace\! =\! \rho_{u}\eta_{k}\xi_{r}\xi_{t}\!\Bigg[\!\sum\limits_{m=1}^{M}\!\!\left(\mathbb{E}\!\left\lbrace{\!\lvert g_{mk}\hat{g}_{mk}^{*}\rvert}^2\right\rbrace\!-\!\gamma_{mk}^{{\prime}^{2}}\right)
\!+\!\sum\limits_{m=1}^{M}\!\sum\limits_{n\neq m}^{M}\!\!\left(\mathbb{E}\left\lbrace{ g_{mk}\hat{g}_{mk}^{*}g_{nk}^{*}\hat{g}_{nk}}\right\rbrace\! - \!\gamma_{mk}^{\prime
}\gamma_{nk}^{\prime}\right)\!\!\Bigg].
\end{equation}

Upper bounds for $\mathbb{E}\!\left\lbrace{\!\lvert g_{mk}\hat{g}_{mk}^{*}\rvert}^2\right\rbrace$ and $\mathbb{E}\left\lbrace{ g_{mk}\hat{g}_{mk}^{*}g_{nk}^{*}\hat{g}_{nk}}\right\rbrace$ are as follows
\begin{equation}\label{e1}
\begin{split}
\mathbb{E}\!\left\lbrace{\!\lvert g_{mk}\tilde{g}_{mk}^{*}\rvert}^2\right\rbrace &=\! \mathbb{E}\!\left\lbrace{\!\lvert g_{mk}\hat{g}_{mk}^{*}\! +\! g_{mk}q_{p,mk}^{*}\rvert}^2\!\right\rbrace\!\! = \!\mathbb{E}\!\left\lbrace{\!\lvert g_{mk}\hat{g}_{mk}^{*}\rvert}^2\right\rbrace \!+ \! \mathbb{E}\!\left\lbrace{\!\lvert g_{mk}q_{p,mk}^{*}\rvert}^2\right\rbrace\! +\! \mathbb{E}\!\left\lbrace{\!\lvert g_{mk}\rvert}^{2}\hat{g}_{mk}^{*}q_{p,mk}\!\right\rbrace \\
& + \!\!\mathbb{E}\!\left\lbrace{\!\lvert g_{mk}\rvert}^{2}\hat{g}_{mk}q_{p,mk}^{*}\right\rbrace\!\! \underset{\text{(a)}}\geq\!\!  \mathbb{E}\!\left\lbrace{\!\lvert g_{mk}\hat{g}_{mk}^{*}\rvert}^2\right\rbrace\! +\! \mathbb{E}\!\left\lbrace{\!\lvert g_{mk}q_{p,mk}^{*}\rvert}^2\!\right\rbrace\! \underset{\text{(b)}}\geq \! \mathbb{E}\!\left\lbrace{\!\lvert g_{mk}\hat{g}_{mk}^{*}\rvert}^2\!\right\rbrace\! +\! Q_{p,mk}\beta_{mk}\\
&+\!Q_{p,mk}^{2} \Rightarrow \mathbb{E}\!\left\lbrace{\!\lvert g_{mk}\hat{g}_{mk}^{*}\rvert}^2\!\right\rbrace\! \leq \mathbb{E}\!\left\lbrace{\!\lvert g_{mk}\tilde{g}_{mk}^{*}\rvert}^2\right\rbrace - Q_{p,mk}\beta_{mk} - Q_{p,mk}^{2},
\end{split}
\end{equation}
where, (a) is due to $\mathbb{E}\!\left\lbrace{\!\lvert g_{mk}\rvert}^{2}\hat{g}_{mk}^{*}q_{p,mk}\right\rbrace\geq 0, \ \mathbb{E}\!\left\lbrace{\!\lvert g_{mk}\rvert}^{2}\hat{g}_{mk}q_{p,mk}^{*}\right\rbrace\geq 0$, and (b) results from $\mathbb{E}\!\left\lbrace{\!\lvert g_{mk}q_{p,mk}^{*}\rvert}^2\!\right\rbrace \geq \mathbb{E}\!\left\lbrace{\!\lvert \hat{g}_{mk}q_{p,mk}^{*}\rvert}^2\!\right\rbrace + \mathbb{E}\!\left\lbrace{\!\lvert e_{mk}\rvert}^2\!\right\rbrace\mathbb{E}\!\left\lbrace{\!\lvert q_{p,mk}^{*}\rvert}^2\!\right\rbrace + \mathbb{E}\!\left\lbrace{\!\lvert q_{mk}q_{p,mk}^{*}\rvert}^2\!\right\rbrace = Q_{p,mk}\beta_{mk}+\!Q_{p,mk}^{2}$, where we have assumed $g_{mk} = \underbrace{\hat{g}_{mk} + q_{p,mk}}_{\tilde{g}_{mk}} + e_{mk}$ in which $e_{mk}\sim\mathcal{CN}\left(0, \beta_{mk}-\gamma_{mk}\right)$ accounts for estimation error which is uncorrelated with $\tilde{g}_{mk}$ but not independent. Noting that, the inequality is due to the independency assumption.

Similarly, we have the following upper-bound for the second expectation in (\ref{BUup}),
\begin{equation}\label{e2}
\begin{split}
\mathbb{E}\!\left\lbrace g_{mk}\tilde{g}_{mk}^{*}g_{nk}^{*}\tilde{g}_{nk} \right\rbrace &=\! \mathbb{E}\!\left\lbrace g_{mk}(\hat{g}_{mk}^{*}+q_{p,mk}^{*})g_{nk}^{*}(\hat{g}_{nk}+q_{p,nk})\!\right\rbrace\!\! \geq \!\mathbb{E}\!\left\lbrace\!g_{mk}\hat{g}_{mk}^{*}g_{nk}^{*}\hat{g}_{nk}\right\rbrace \!+ \! Q_{p,mk}Q_{p,nk}\\
& \Rightarrow \mathbb{E}\!\left\lbrace\!g_{mk}\hat{g}_{mk}^{*}g_{nk}^{*}\hat{g}_{nk}\right\rbrace\! \leq \mathbb{E}\!\left\lbrace g_{mk}\tilde{g}_{mk}^{*}g_{nk}^{*}\tilde{g}_{nk} \right\rbrace - Q_{p,mk}Q_{p,nk},
\end{split}
\end{equation}
$\mathbb{E}\!\left\lbrace g_{mk}\tilde{g}_{mk}^{*}g_{nk}^{*}\tilde{g}_{nk} \right\rbrace$ and $\mathbb{E}\!\left\lbrace{\!\lvert g_{mk}\tilde{g}_{mk}^{*}\rvert}^2\right\rbrace$ could be computed in the same way as in Appendix A, then by replacing $\mathbb{E}\!\left\lbrace\!g_{mk}\hat{g}_{mk}^{*}g_{nk}^{*}\hat{g}_{nk}\right\rbrace$ and $\mathbb{E}\!\left\lbrace{\!\lvert g_{mk}\hat{g}_{mk}^{*}\rvert}^2\!\right\rbrace$ with their upper-bounds (\ref{e1}) and (\ref{e2}) in equation (\ref{BUup}) the following upper-bound can be obtained
\begin{equation*}
\begin{split}
  \mathbb{E}\left\lbrace{\!\lvert\text{BU}_{k}\rvert}^2\right\rbrace &= \rho_{u}\eta_{k}\xi_{r}\xi_{t}\Bigg[\sum\limits_{m=1}^{M}\gamma_{mk}^{\prime}\beta_{mk} + \frac{1-\xi_{t}}{\tau\xi_{t}}\left(\sum\limits_{m=1}^{M}\gamma_{mk}\right)^2 + \rho_{p}(1-\xi_{r})\sum\limits_{m=1}^{M}\lambda_{mk}^{2}\beta_{mk}^{2}\\
  &+2\left(\sum\limits_{m=1}^{M}Q_{p,mk}\right)\left(\sum\limits_{m=1}^{M}\gamma_{mk}^{\prime}\right)\Bigg],
  \end{split}
\end{equation*}
One can obtain upper-bounds for other parts of the interference similarly,
\begin{equation*}
\begin{split}
  \mathbb{E}\left\lbrace{\!\lvert\text{IUI}_{kk^{\prime}}\rvert}^2\right\rbrace &= \rho_{u}\eta_{k^{\prime}}\xi_{r}\xi_{t}\Bigg[\sum\limits_{m=1}^{M}\gamma_{mk}\beta_{mk^{\prime}} + \frac{1-\xi_{t}}{\tau\xi_{t}}\left(\sum\limits_{m=1}^{M}\gamma_{mk}\frac{\beta_{mk^{\prime}}}{\beta_{mk}}\right)^2 + \rho_{p}(1-\xi_{r})\sum\limits_{m=1}^{M}\lambda_{mk}^{2}\beta_{mk^{\prime}}^{2}\\
  &-\sum\limits_{m=1}^{M}Q_{p,mk}Q_{p,mk^{\prime}}\Bigg],
  \end{split}
\end{equation*}
\begin{equation*}
\begin{split}
\hspace{-1cm}\mathbb{E}\left\lbrace{\!\lvert\text{THI}_{kk^{\prime}}\rvert}^2\right\rbrace &= \rho_{u}\eta_{k^{\prime}}\xi_{r}(1-\xi_{t})\Bigg[\sum\limits_{m=1}^{M}\gamma_{mk}^{\prime}\beta_{mk^{\prime}} + \left(\boldsymbol{\varphi}_{k}^{H}\boldsymbol{\varphi}_{k^{\prime}} + \frac{1-\xi_{t}}{\tau\xi_{t}}\right)\!\!\left(\sum\limits_{m=1}^{M}\gamma_{mk}\frac{\beta_{mk^{\prime}}}{\beta_{mk}}\right)^2 \\
&+ \rho_{p}(1-\xi_{r})\sum\limits_{m=1}^{M}\lambda_{mk}^{2}\beta_{mk^{\prime}}^{2}-\sum\limits_{m=1}^{M}Q_{p,mk}Q_{p,mk^{\prime}}\Bigg],
  \end{split}
\end{equation*}
\begin{equation*}
\begin{split}
\hspace{-1cm}\mathbb{E}\left\lbrace{\!\lvert\text{RHI}_{k}\rvert}^2\right\rbrace =& \sum\limits_{k^{\prime} = 1}^{K}\rho_{u}\eta_{k^{\prime}}(1-\xi_{r})\Bigg[\sum\limits_{m=1}^{M}\gamma_{mk}^{\prime}\beta_{mk^{\prime}} + \rho_{p}\xi_{r}\Big(\tau\xi_{t}\boldsymbol{\varphi}_{k}^{H}\boldsymbol{\varphi}_{k^{\prime}} + (1-\xi_{t})\Big)\!\!\sum\limits_{m=1}^{M}\lambda_{mk}^{2}\beta_{mk^{\prime}}^{2} \\
& + \rho_{p}(1-\xi_{r})\sum\limits_{m=1}^{M}\lambda_{mk}^{2}\beta_{mk^{\prime}}^{2}-\rho_{p}(1-\xi_{r})\sum\limits_{m=1}^{M}Q_{p,mk}Q_{p,mk^{\prime}}\Bigg].
  \end{split}
\end{equation*}
\section*{Appendix C}\label{app:ECF_UpperBound}
\vspace{-.41cm}
Noting that $\mathbb{E}\{\tilde{g}_{mk}^{*}g_{mk}\} = \gamma_{mk}$, consequently $\boldsymbol{b}_{k} = \sqrt{\rho_{u}\eta_{k}\xi_{r}\xi_{t}}[\gamma_{1k}, \ \gamma_{2k}, \ ..., \ \gamma_{mk}, \ ..., \ \gamma_{Mk}]$. For the off-diagonal elements of $\mathcal{K}_{\boldsymbol{z}_k}$ we have

\begin{equation}\label{OffDiag}
  \begin{split}
     \mathcal{K}_{\boldsymbol{z}_{k}}[n,m]\!\! =& \mathcal{K}_{\boldsymbol{z}_{k}}[m,n] = \mathbb{E}\{\boldsymbol{z}_{k}[m,1]\boldsymbol{z}_{k}^{*}[n,1]\} = \xi_{r}\xi_{t}\eta_{k}\rho_{u}\Big[\mathbb{E}\{\tilde{g}_{mk}^{*}g_{mk}\tilde{g}_{nk}g_{nk}^{*}\}-\gamma_{mk}\gamma_{nk}\Big]\\
       +& \xi_{r}\xi_{t}\rho_{u}\sum\limits_{k^{\prime}\neq k}^{K}\eta_{k^{\prime}}\mathbb{E}\{\tilde{g}_{mk}^{*}g_{mk^{\prime}}\tilde{g}_{nk}g_{nk^{\prime}}^{*}\} + \xi_{r}\sum\limits_{k^{\prime
       } = 1}^{K}\mathbb{E}\{\tilde{g}_{mk}^{*}g_{mk^{\prime}}w_{t,k^{\prime}}\tilde{g}_{nk}g_{nk^{\prime}}^{*}w_{t,k}^{*}\}\\
        =& \xi_{r}\xi_{t}\eta_{k}\rho_{u}\Big[\!\gamma_{mk}\gamma_{nk}\!+\!\frac{1-\xi_{t}}{\tau\xi_{t}}\gamma_{mk}\gamma_{nk}\!-\!\gamma_{mk}\gamma_{nk}\Big]
       \!+\!\rho_{u}\xi_{r}\xi_{t}\!\!\sum\limits_{k^{\prime}\neq k}^{K}\!\!\eta_{k^{\prime}}\frac{1-\xi_{t}}{\tau\xi_{t}}\frac{\beta_{nk^{\prime}}\beta_{mk^{\prime}}}{\beta_{mk}\beta_{nk}}\gamma_{mk}\gamma_{nk}\\
       +&\rho_{u}\xi_{r}(1-\xi_{t})\sum\limits_{k^{\prime}=1}^{K}\eta_{k^{\prime}}\frac{1-\xi_{t}}{\tau\xi_{t}}\frac{\beta_{nk^{\prime}}\beta_{mk^{\prime}}}{\beta_{mk}\beta_{nk}}\gamma_{mk}\gamma_{nk}
       = \rho_{u}\xi_{r}\frac{1-\xi_{t}}{\tau\xi_{t}}\sum\limits_{k^{\prime}= 1}^{K}\eta_{k^{\prime}}\frac{\beta_{nk^{\prime}}\beta_{mk^{\prime}}}{\beta_{mk}\beta_{nk}}\gamma_{mk}\gamma_{nk}.
  \end{split}
\end{equation}
For the diagonal elements after some more mathematical calculations one can show that
\begin{equation}\label{Diag}
  \begin{split}
     \mathcal{K}_{\boldsymbol{z}_{k}}[m,m] \!\!=& \mathbb{E}\{{\lvert\boldsymbol{z}_{k}[m,1]\rvert}^{2}\} = \rho_{u}\eta_{k}\xi_{r}\xi_{t}\Big[\mathbb{E}\{{\lvert \tilde{g}_{mk}^{*}g_{mk}\rvert}^2\}\Big] + \rho_{u}\xi_{r}\xi_{t}\sum\limits_{k^{\prime}\neq k}^{K}\eta_{k^{\prime}}\mathbb{E}\{{\lvert \tilde{g}_{mk}^{*}g_{mk^{\prime}}\rvert}^2\}\\
     +& \rho_{u}\xi_{r}(1-\xi_{t})\sum\limits_{k^{\prime}=1}^{K}\eta_{k^{\prime}}\mathbb{E}\{{\lvert \tilde{g}_{mk}^{*}g_{mk^{\prime}}\rvert}^2\} + \mathbb{E}\{{\lvert \tilde{g}_{mk}^{*}w_{r,m}\rvert}^2\} + \mathbb{E}\{{\lvert \tilde{g}_{mk}^{*}n_{m}\rvert}^2\} + \mathbb{E}\{{\lvert q_{mk}\rvert}^2\}\\
     =& \rho_{u}\eta_{k}\xi_{r}\xi_{t}\Big[\gamma_{mk}\beta_{mk}-\frac{1}{\tau}\gamma_{mk}^{2}+\rho_{p}\lambda_{mk}^{2}\beta_{mk}^{2}\Big] + \rho_{u}\xi_{r}\xi_{t}\sum\limits_{k^{\prime}\neq k}^{K}\eta_{k^{\prime}}\Big(\gamma_{mk}\beta_{mk^{\prime}}-\frac{1}{\tau}\frac{\beta_{mk^{\prime}}^{2}}{\beta_{mk}^{2}}\gamma_{mk}^{2}\\ +& \rho_{p}\lambda_{mk}^{2}\beta_{mk^{\prime}}^{2}\Big)
     \!+\! \rho_{u}\xi_{r}(1\!-\!\xi_{t})\!\!\sum\limits_{k^{\prime}=1}^{K}\!\!\eta_{k^{\prime}}\Big(\gamma_{mk}\beta_{mk^{\prime}}\!+\!(\boldsymbol{\varphi}_{k}^{H}\!\boldsymbol{\varphi}_{k^{\prime}}\!-\!\frac{1}{\tau})\!\frac{\beta_{mk^{\prime}}^{2}}{\beta_{mk}^{2}}\gamma_{mk}^{2} \!+ \!\rho_{p}\lambda_{mk}^{2}\beta_{mk^{\prime}}^{2}\!\Big)\\
     +& \rho_{u}(1\!-\!\xi_{r})\!\!\sum\limits_{k^{\prime}=1}^{K}\!\!\eta_{k^{\prime}}\Big(\gamma_{mk}\beta_{mk^{\prime}}\!+\!(\boldsymbol{\varphi}_{k}^{H}\!\boldsymbol{\varphi}_{k^{\prime}}\!-\!\frac{1}{\tau})\frac{\beta_{mk^{\prime}}^{2}}{\beta_{mk}^{2}}\gamma_{mk}^{2} \!+ \!\rho_{p}\lambda_{mk}^{2}\beta_{mk^{\prime}}^{2}\!\Big)\! + \!N\gamma_{mk}\! + \!Q_{mk}.
  \end{split}
\end{equation}
By further simplifications the result in equation (\ref{eqn:RateWCF}) is obtained.
\bibliography{Refrences}
\bibliographystyle{ieeetr}

\end{document}